\documentclass{article}
\usepackage{fullpage}
\usepackage[italian,francais,german,english]{babel}
\usepackage{latexsym}
\usepackage{amsmath,amsthm,bbm,mathrsfs}
\usepackage{amssymb}
\def\d{{\mathrm d}}
\def\RR{{\mathbb R}}

\def\NN{{\mathbb N}}

\numberwithin{equation}{section}
\newcommand{\beq}{\begin{equation}}
\newcommand{\norm}[1]{\lVert #1 \rVert}
\newcommand{\eeq}{\end{equation}}
\newcommand{\beqn}{\begin{eqnarray}}
\newcommand{\eeqn}{\end{eqnarray}}
\newcommand{\bal}{\begin{align}}
\newcommand{\eal}{\end{align}}
\newcommand{\xt}{\widetilde{X}}
\renewcommand{\Re}{\mathrm{Re}\,}
\renewcommand{\Im}{\mathrm{Im}\,}
\newcommand{\scalar}[2]{\langle{#1} \mspace{2mu}, {#2}\rangle}
\def\CC{{\mathbb C}}
\def\({\left(}
\def\){\right)}
\def\L{\varLambda}

\def\11{{\mathbbmss 1}}

\def\i{{\mathrm i}}
\def\e{{\mathrm e}}

\def\eps{{\varepsilon}}
\def\mez{{-\frac{1}{2}}}
\def\pez{{\frac{1}{2}}}
\def\mdz{{-\frac{3}{2}}}
\def\pdz{{\frac{3}{2}}}
\def\mfv{{-\frac{1}{2}-\delta}}

\def\h{{h^{X_T}}}
\def\W{{W^{X_T}}}
\newcommand{\avg}[1]{\langle #1 \rangle}
\def\x{{\avg{x}}}

\newtheorem{thm}{Theorem}[section]
\newtheorem{lem}[thm]{Lemma}
\newtheorem*{lem*}{Lemma}
\newtheorem{cor}[thm]{Corollary}
\newtheorem{prop}[thm]{Proposition}

\begin{document}

\title{Some Hamiltonian Models of Friction II}
\author{Daniel Egli\footnote{danegli@itp.phys.ethz.ch} and Zhou Gang\footnote{zhougang@itp.phys.ethz.ch}}

\maketitle
\setlength{\leftmargin}{.1in}
\setlength{\rightmargin}{.1in}
\normalsize \vskip.1in
\setcounter{page}{1} \setlength{\leftmargin}{.1in}
\setlength{\rightmargin}{.1in}
\centerline{$^{\ast,\dagger}$Institute for Theoretical Physics, ETH Zurich, CH-8093, Z\"urich, Switzerland}
\date
\abstract{In the present paper we consider the motion of a very heavy tracer particle in a medium of a very dense, non-interacting Bose gas. We prove that, in a certain mean-field limit, the tracer particle will be decelerated and come to rest somewhere in the medium. Friction is caused by emission of Cerenkov radiation of gapless modes into the gas. Mathematically, a system of semilinear integro-differential equations, introduced in \cite{Froehlich10}, describing a tracer particle in a dispersive medium is investigated, and decay properties of the solution are proven. This work is an extension of \cite{Froehlich102}; it is an extension because no weak coupling limit for the interaction between tracer particle and medium is assumed. The technical methods used are dispersive estimates and a contraction principle.}

\section{Introduction}
In \cite{Froehlich10} a model of quantum friction is introduced. A tracer particle is coupled to a bath of identical bosons. It is heuristically motivated that the regime of a very dense but weakly interacting bose gas, and heavy tracer particle, corresponds to a classical limit, and the model reduces to a classical Hamiltonian system for $(X,P)$, the position and momentum of the tracer particle, and $\beta(x)$, the fluctuation from the mean density $\rho_0$ of the bosons.

The resulting equations are
\begin{align}
\dot{X_t}=&\frac{P_t}{M},\quad\quad\\
\dot{P_t}=&-\nabla_{X}V(X_t)-g\int\nabla_{X}W(x-X_t)(|\beta_t(x)|^2
+2\sqrt{\frac{\rho_0}{g^2}}Re\beta_t(x)) dx,    \label{XPeqns'}\\
i\dot{\beta}_t(x)=&(-\frac{1}{2m}\Delta+gW(x-X_t))
\beta_t(x)+\sqrt{\rho_0}W(x-X_t)                  \nonumber\\
+&\kappa[\phi *(|\beta_t|^2+2\sqrt{\frac{\rho_0}{g^2}}
\Re\beta_t)](x)(\beta_t(x)+\sqrt{\frac{\rho_0}{g^2}})\,, \label{betaeqn}
\end{align}
where $V$ is an external potential affecting only the tracer particle, $W$ is a two-body potential modelling the interaction between tracer particle and the medium, $\phi$ is a two-body potential modelling the interaction between medium particles, and $g,\kappa$ are coupling constants. See \cite{Froehlich10,Froehlich102} for a detailed description of the model.

These equations are Hamiltonian with Hamilton functional
\begin{align*}
H(X,P,\beta,\bar{\beta}) =& \frac{P^2}{2M}+V(X)+\int \left[\frac{1}{2m}|\nabla\beta(x)|^2+gW(x-X)(|\beta(x)|^2+2\sqrt{\frac{\rho_0}{g^2}}\Re\beta(x))\right]\d x\\
+&\frac{\kappa}{2} \iint(|\beta(x)|^2+2\sqrt{\frac{\rho_0}{g^2}}\Re\beta(x))\phi(x-y)(|\beta(y)|^2+2\sqrt{\frac{\rho_0}{g^2}}\Re\beta(y))\d x\d y\,,
\end{align*}
and the standard symplectic form $P'\cdot X-P\cdot X'+2\i\,\Im \int\bar{\beta}\beta'$.\\
In \cite{Froehlich102}, the model of a non-interacting medium without external forcing ($\kappa=0, V=0$) and a tracer particle that is weakly coupled to the medium ($g\to 0$) is studied. In the present paper, we go one step further and restore the full coupling of the tracer particle to the medium, that is, we consider the following parameter regime:
$$\kappa=0\ \text{and}\ g\neq 0.$$
The case of an interacting medium ($\kappa\neq 0$) will hopefully be treated in a forthcoming paper.

The equations take the form
\begin{align}\label{equations}
\dot{X}_t&=\frac{P_t}{M}\nonumber\\
\dot{P}_t&=-\partial_xV(X_t)-g\int_{\RR^3}\partial_xW^{X_t}\(|\beta_t(x)|^2+2\sqrt{\frac{\rho_0}{g^2}}\Re \beta_t(x)\)\d x\\
\i\dot{\beta}_t(x)&=h^{X_t}\beta_t+\sqrt{\rho_0}W^{X_t}\,,\nonumber
\end{align}
where $h^{X_t}:=-\frac{1}{2m}\Delta+gW^{X_t}$ and $W^{X_t}(x):=W(x-X_t)$.

In the main part of the paper we consider the cases the external potential $V$ vanishes, and prove that the tracer particle experiences friction and is decelerated to a full stop. We prove a lower bound for the strength of this friction mechanism, namely $|P_t|\leq c t^{-1-\eps}$, $t\to\infty$, for some explicit $\eps>0$ depending on the initial conditions. At large times, the medium is shown to exhibit the expected behavior: It forms a ``splash'' that follows the motion of the tracer particle. Remarkably, even though initial conditions $\beta_0$ can be chosen to be very small (in $L^2$-sense), the splash that the medium forms eventually is not square integrable. This is a consequence of the fact that we chose the medium to be non-interacting. This fact is also responsible for making it difficult to ``guess'' the right asymptotic behaviour of $|P_t|$ on a heuristic level. See \cite{Froehlich10,Froehlich102} for a more thorough discussion.

The second author, together with his collaborators, considered in \cite{Froehlich102} the problem with $\kappa=g=0$. They found completely analogous results. Nevertheless, our findings are interesting in their own right as we treat a particle coupled fully to the medium (as opposed to a weak coupling limit), which is usually a much harder problem. The main technical difference is that the generator of time evolution of the reservoir, $h^{X_t}=-\Delta+gW^{X_t}$, depends on time, for $g\neq 0$, through the position, $X_t$, of the particle.

The remainder of the paper is organized as follows. In section 2, we present the main mathematical result. In section 3, we analyze the local well-posedness. In section 4, we recast the equations in a more convenient form; in particular, we expand the time-dependent propagator around its value at some fixed large time, and we split the equation for $P_t$ into linear and non-linear parts. In section 5, we apply a contraction principle to prove the existence of the solution $P_t$ with the desired decay, and in section 6 we prove the main theorem. Technical proofs of various propositions used along the way have been relegated to the appendix.

\section{Main theorem}
In order to be able to state a precise theorem, introduce the continuous, monotonically increasing function $\Omega:(-\infty,1)\to \RR^+$,
\begin{align}\label{eq:defOmega}
\Omega(\delta):=\frac{1}{\pi}\int_0^1\frac{1}{1+(1-r)^\pez}(1-r)^\mez\(\frac{1}{1-2\delta}(r^\mez-r^{-\delta})+r^{\pez-\delta}\)\d r\,.
\end{align}
By numerical simulation we find that there exists a constant $\delta^*\simeq 0.66$ such that
\begin{align*}
\Omega(\delta^*)=1.
\end{align*}
Moreover, for any constant $\delta<\delta^*$, we have
$$\Omega(\delta)<1.$$
For the system of equations (\ref{equations}) we prove the following main theorem,
\begin{thm}\label{thm:main}Suppose that in \eqref{equations} the external potential vanishes, $V\equiv 0$, and the potential $W$ is smooth, spherically symmetric, decays rapidly at $|x|=\infty$, and satisfies
  \begin{align*}
   |\widehat{W}(0)|\neq 0\,.
  \end{align*}
Then, for any $\delta\in I:=(\frac{1}{2},\delta^*)$ there exists a $g_0>0$ and an $\eps_0>0$ such that if $0\leq g\leq g_0$ and $\norm{\x^5\beta_0}_2,|P_0|\leq\eps_0$ and $\norm{\x^3\partial_x\beta_0}_2<\infty$ then
  \begin{align}\label{momentumdecay}
    |P_t|\leq ct^{-\frac{1}{2}-\delta} \;\textrm{as $t\to\infty$,}
  \end{align}
and
\begin{align}\label{deltadecay}
\lim_{t\to \infty}\norm{\x^{-3}(\beta_t+\sqrt{\rho_0}(h^{X_t})^{-1}W^{X_t})}_2=0\,.
\end{align}
In particular, the particle comes to rest after a finite distance: There is an $X_\infty\in\RR^3$ such that $X_t\to X_\infty$, and
\begin{align*}
 \beta_t \to -\sqrt{\rho_0}(h^{X_\infty})^{-1}W^{X_\infty}\notin L^2(\RR^3)\,.
\end{align*}
\end{thm}
The theorem will be proved in section \ref{mainproof}.

Now we present the main difficulties in the proof and the strategies of overcoming them.
Similar to what was proved in \cite{Froehlich102}, we start with decomposing the equation for $\dot{P}_t$ into a linear and a non-linear part. Part of the linear equations can be solved explicitly, and we use the solution to rewrite the equation for $P_t$ in terms of this solution and the non-linear part. Since we expect that the momentum $P_t$ decays for large times $t$ it is reasonable to assume that eventually the dynamics is dominated by the linear part. The detailed knowledge of the decay properties of the solution to the linear part and standard dispersive estimates enable us to use a contraction principle to establish the claim. It is recommended that the reader consult \cite{Froehlich102} for a more precise outline of the general strategy which is almost identical to the present case.

There is one major technical difference to the model studied in \cite{Froehlich102}, namely that the generator of time evolution, $h^{X_t}=-\Delta+gW^{X_t}$, depends on time through the position, $X_t$, of the particle. Mathematically, this makes it more involved to cancel various terms by symmetry considerations, and, as an additional complication, the generator of translations, $\partial_x$, no longer commutes with $h^{X_t}$. We treat this as follows. Since we expect that the particle will come to rest at some $X_\infty\in\RR^3$, we expand the propagator $U(t,s)$ gererated by $ h^{X_t}=-\Delta+gW^{X_t}$ around the ``instantaneous'' propagator, $\e^{-\i h^{X_t} t}$, at some large time $t$ where
\begin{align*}
\e^{-\i h^{X_t} t}=\e^{-\i h^{X_T} t}\big|_{t=T}
\end{align*}
is to be understood. By Duhamel's principle we obtain
\begin{align*}
U(t,0)=\e^{-\i h^{X_t} t}-\i\int_0^t\e^{-\i h^{X_t}(t-s)}(X_s-X_t)\cdot\partial_xW^{X_t}\e^{-\i h^{X_t} s}\d s+\dots\,.
\end{align*}

To facilitate later discussions we rescale the equation such that
$$2m=1, \ |\widehat{W}(0)|=1.$$
\section{The local well-posedness}
In this section we discuss the local well-posedness of solutions to equation system ~\eqref{equations}.

Apply Duhamel's principle on the last equation of ~\eqref{equations} to obtain
\begin{align}\label{eq:localEx}
\beta_{t}=U(t,0)\beta_0-\i\sqrt{\rho_0}\int_{0}^{t}U(t,s)W^{X_{s}}\ ds\,,
\end{align} where $U(t,s)$ is the propagator generated by the operators $h^{X_{\tau}}=-\Delta+gW^{X_{\tau}},\ \tau\in [s,t]$. Since the right hand side does not depend on $\beta_{\cdot},$ one can see that for any given trajectory $X_{\cdot}$ there exists a solution $\beta_{\cdot}$, with $\beta_{t}\in L^{2}(\mathbb{R}^3)$ for any time $t\in [0,\infty).$

To simplify the problem we plug ~\eqref{eq:localEx} into ~\eqref{equations}. Based on the discussion above, the local existence of the solution is transformed into the local existence of the trajectory and momentum. The latter can be achieved by a standard iteration technique. The proof is simple, but tedious. Hence we omit the details here.

Furthermore, we observe that if $\beta_0\equiv 0$ and $P_0=0$ then $X_t=X_0$ and $P_t=0$ is a solution. This, together with the local well-posedness, implies that a small solution exists in a large time interval. This is the content of the next theorem.
\begin{thm}\label{thm:localwp}The equation (\ref{equations}) is locally well-posed: If $P_0\in\RR^3$ and $\x^3\beta_0\in L^2(\RR^3)$ then there exists a time $T_{\rm loc}=T_{\rm loc}(|P_0|,\norm{\x^3\beta_0}_2)$ such that $|P_t|<\infty$ for any time $t\leq T_{\rm loc}$. Moreover, for any $T_{\rm loc}>0$ there exists an $\eps_0(T_{\rm loc})$ such that if $|P_0|,\norm{\x^3\beta_0}_2\leq \eps_0(T_{\rm loc})$ then $P$ satisfies the estimate
  \begin{align}\label{eq:localdecay}
    |P_t|\leq T_{\rm loc}^{-2} \quad t\in [0,T_{\rm loc}]\,.
  \end{align}
\end{thm}

\section{Reformulation of the problem}
We begin with presenting the main difficulties and the strategies in overcoming them. Equation ~\eqref{eq:localEx}, obtained from the last equation in ~\eqref{equations}, does not help directly in our analysis. To illustrate the difficulty we plug ~\eqref{eq:localEx} into the right hand side of the second equation of ~\eqref{equations}. One of the terms we obtain is
\begin{align*}
\Psi_t:=\rho_0 \Re \i\langle \partial_{x}W^{X_{t}}, \int_{0}^{t} U(t,s)W^{X_{s}}\rangle\ \d s.
\end{align*} In order to show $P_t\rightarrow 0$ as $t\rightarrow \infty$, we have to prove that this term is small. To this end, we will prove that $|X_t-X_s|$ is small, which yields
\begin{align}
U(t,s)W^{X_{s}}=\e^{\i(t-s)h^{X_{t}}}W^{X_{t}}+O(|X_t-X_s|).
\end{align} Put this into the expression for $\Psi$ to find that the contribution of the first term is zero because $W$ is a spherically symmetric function.
To make this rigorous and to further control the terms in $O(|X_t-X_s|)$ we Duhamel-expand the term $U(t,s)W^{X_{s}}$ around $\e^{\i(t-s)h^{X_{t}}}W^{X_{t}}$ to a certain order.

However, as it turns out, it is not convenient to work on the term $U(t,s)W^{X_{s}}$ as some of the information becomes hard to see. In what follows we present a different approach with the Duhamel expansion still being the underlying idea.

Since the generator $h^{X_t}$ depends on the position $X_t$ of the particle, we expand it around its value at a position $X_T$ for any fixed time $T>0$. Define $\bar{\beta}^X:=-\sqrt{\rho_0}(h^X)^{-1}W^X$ and introduce a new function $\delta_t=\delta_{t,T}$ by
\begin{align}\label{eq:decomp}
  \beta_t-\bar{\beta}^{X_T}-\sqrt{\rho_0}\sum_{|\alpha|=1}^{N_0}\frac{1}{\alpha!}(X_t-X_T)^\alpha\partial_x^\alpha(\h)^{-1}\W
  =:\delta_t\,.
\end{align} Here the positive integer $N_0$ is defined as
\begin{align}\label{eq:N0}
N_0:=\min\{n\in\NN:(n+1)(\delta-\frac{1}{2})\geq \frac{3}{2}+\delta\}.
\end{align} Recall the constant $\delta$ in Theorem ~\ref{thm:main}. The motivation for choosing $N_0$ as above is that if we can prove $|P_t|=O(t^{-\frac{1}{2}-\delta})$, then $|X_t-X_T|^{N_0+1}=O(t^{-\frac{3}{2}})$.

Now, $\delta_t$ satisfies the equation
\begin{align}\label{eq:delta}
\i\dot{\delta}_t&=\h\delta_t+g(W^{X_t}-\W)\delta_t-\i\frac{\sqrt{\rho_0}}{M}P_t\cdot\sum_{|\alpha|=1}^{N_0}\frac{1}{\alpha!}\alpha(X_t-X_T)^{\alpha-1}\partial_x^\alpha(\h)^{-1}\W -G_1\nonumber\\
\delta_0&=\beta_0-\bar{\beta}^{X_T}-\sqrt{\rho_0}\sum_{|\alpha|=1}^{N_0}\frac{1}{\alpha!}(X_0-X_T)^\alpha\partial_x^\alpha(\h)^{-1}\W\,,
\end{align}
where $\alpha X^{\alpha-1}$ means the vector $X=(\alpha_1X^{(\alpha_1-1,\alpha_2,\alpha_3)},\alpha_2X^{(\alpha_1,\alpha_2-1,\alpha_3)},\alpha_3X^{(\alpha_1,\alpha_2,\alpha_3-1)})$, and the term $G_1$ is defined as
\begin{align}\label{eq:G1}
G_1:=h^{X_t}r_{N_0}\,,
\end{align}
with $r_{N_0}$ defined by
\begin{align}\label{eq:barbeta}
\bar{\beta}^{X_t}=:\bar{\beta}^{X_T}+\sqrt{\rho_0}\sum_{|\alpha|\leq N} \frac{1}{\alpha!}(X_t-X_T)^\alpha\partial_x^\alpha(\h)^{-1}\W+r_{N_0}\,,
\end{align}
and estimated in the following lemma. Define an estimating function $\mu:\ \mathbb{R}^{+}\rightarrow \mathbb{R}^{+}$ by
\begin{align}\label{majorant}
  \mu(t):=\max_{0\leq s\leq t}(1+s)^{\frac{1}{2}+\delta}|P_s|\,.
\end{align} and recall the definition of $N_0$ in ~\eqref{eq:N0}.
\begin{lem}\label{lem:taylor}
If $\mu(t)\leq 1$ in the interval $[0,T],$ then in the same interval the function $r_{N_0}$ in ~\eqref{eq:barbeta} satisfies the estimate
   \begin{align}\label{eq:rN}
   \norm{\x^3h^{X_t}r_{N_0}}_2\leq C_{N_0}|X_t-X_T|^{N_0+1}.
   \end{align}
   \begin{align}\label{eq:rN2}
   \|\langle x\rangle^{-3}r_{N_0}\|_2 \leq C_{N_0}|X_t-X_T|^{N_0+1}.
   \end{align}
\end{lem}
\begin{proof}
Use the fact $r_{N_0}(s)$ is the remainder term in the Taylor expansion of $(h^{X_s})^{-1}W^{X_s}$ to write the expression as
\begin{align*}
r_{N_0}(s)=(-1)^{N_0+1}\int_t^s\int_t^{s_1}\dots\int_t^{s_{N_0}}\partial_x^{j_1}\dots\partial_x^{j_{N_0+1}}(h^{X_{s_{N_0+1}}})^{-1}W^{X_{s_{N_0+1}}}\dot{X}_{s_{N_0+1}}^{j_{N_0+1}}\dots \dot{X}_{s_1}^{j_{1}}\d \underline{s}\,.
\end{align*}
The claim follows immediately by Taylor-expanding the function $\bar{\beta}_{t}$ around $\bar{\beta}_{T}$ in the vector variable $X_{t}-X_{T}$. To control the remainder we have used the fact that $(h^X)^{-1}$ is a bounded operator from $L^{2,3}$ to $L^{2,-3}$, and the exponential decay of $W$.
\end{proof}

Using Duhamel's principle we can rewrite $\delta_t$ in the form
  \begin{align}\label{eq:deltaduha1}
  &\delta_t&&=&&\e^{-\i\h t}\delta_0-\i g\int_0^t\e^{-\i\h(t-s)}[W^{X_s}-\W]\delta_s\d s\nonumber\\
&&&&&-\frac{\sqrt{\rho_0}}{M}\sum_{|\alpha|=1}^{N_0}\frac{1}{\alpha!}\int_0^t\e^{-\i\h (t-s)}\partial_x^\alpha(\h)^{-1}\W P_s\alpha(X_s-X_T)^{\alpha-1}\d s
+\i\int_0^t\e^{-\i\h (t-s)}G_1(s)\d s\,.
\end{align}
The function $\delta_t$ admits the following estimate:
\begin{prop}\label{prop:delta}
If $\mu(T)\leq 1$ then for any $\tau\leq T$ we have
  \begin{align}\label{eq:bounddelta}
    \norm{\x^{-3}\delta_\tau}_2\lesssim (1+\tau)^\mez\,.
  \end{align}
\end{prop}
The proposition will be proved in Appendix ~\ref{sec:ControlRem}.

In what follows we derive an equation for $\dot{P}_t$. To this end, we rewrite equation \eqref{eq:deltaduha1} for $\delta_t$ as
\begin{align}\label{eq:deltaduha2}
&\delta_t&&=&&\e^{-\i\h t}\sqrt{\rho_0}(X_0-X_T)\cdot\partial_x(\h)^{-1}\W-\frac{\sqrt{\rho_0}}{M}\int_0^t\e^{-\i\h (t-s)}P_s\cdot\partial_x(\h)^{-1}\W \d s\nonumber\\
&&&&&+\e^{-\i\h t}(\beta_0-\bar{\beta}^{X_T}+\sqrt{\rho_0}\sum_{|\alpha|=2}^{N_0}\frac{1}{\alpha!}(X_0-X_T)^\alpha\partial_x^\alpha(\h)^{-1}\W)\nonumber\\
&&&&&-\i g\int_0^t\e^{-\i\h(t-s)}[W^{X_s}-\W]\delta_s\d s\nonumber\\
&&&&&-\frac{\sqrt{\rho_0}}{M}\sum_{|\alpha|=2}^{N_0}\frac{1}{\alpha!}\int_0^t\e^{-\i\h (t-s)}\partial_x^\alpha(\h)^{-1}\W P_s\alpha(X_s-X_T)^{\alpha-1}\d s
+\i\int_0^t\e^{-\i\h (t-s)}G_1(s)\d s\nonumber\\
&&&=:&&\sum_{n=1}^6D_n(t)\,,
\end{align}
where $D_1$ and $D_2$ will be the main terms (being linear in $P_t$) in the equation for $\dot{P}_t$, whereas $D_3$ through $D_6$ will constitute remainder terms.

Recalling (\ref{equations}) and using $\beta_T=\bar{\beta}^{X_T}+\delta_T$ we thus arrive at the following equation for $\dot{P}_t$, where we evaluate at $t=T$ to effect the cancelations due to spherical symmetry, which is only perfect when all centers agree:
\begin{align*}
\dot{P}_t\big|_{t=T}=&-2\rho_0\Re\scalar{\partial_x W^{X_T}}{\e^{-\i\h T}(X_0-X_T)\cdot\partial_x(\h)^{-1}\W}\nonumber\\
&-2g\sqrt{\rho_0}\Re\scalar{\bar{\beta}^{X_T}\partial_x W^{X_T}}{\e^{-\i\h T}(X_0-X_T)\cdot\partial_x(\h)^{-1}\W}\nonumber\\
&+2\frac{\rho_0}{M}\Re\scalar{\partial_x W^{X_T}}{\int_0^T\e^{-\i\h (T-s)}P_s\cdot\partial_x(\h)^{-1}\W \d s}\nonumber\\
&+2g\frac{\sqrt{\rho_0}}{M}\Re\scalar{\bar{\beta}^{X_T}\partial_x W^{X_T}}{\int_0^T\e^{-\i\h (T-s)}P_s\cdot\partial_x(\h)^{-1}\W \d s}\nonumber\\
&+ R(P,T)\,,
\end{align*}
with $R(P,T)$ defined as
\begin{align}\label{eq:remainder}
R(P,T)=&-2\sqrt{\rho_0}\scalar{(1+\frac{g}{\sqrt{\rho_0}}\bar{\beta}^{X_T})\partial_x\W}{\sum_{n=3}^6D_n}
-g\scalar{\partial_x\W}{|\delta_T|^2}\\
=&\sum_{k=3}^7\tilde{D}_{k}\,,\nonumber
\end{align} where the $\tilde{D}_{k}$ are naturally defined.
By shifting the center of integration and using the spherical symmetry of $W$ the above equation is equivalent to $(k=1,2,3)$
\begin{align*}
\dot{P}^{(k)}_T=&-2\rho_0\Re\scalar{(1+\frac{g}{\sqrt{\rho_0}}\bar{\beta})\partial_{x_k} W}{\e^{-\i h T}(X_0-X_T)_k\partial_{x_k}( h)^{-1} W}\nonumber\\
&+2\frac{\rho_0}{M}\Re\scalar{(1+\frac{g}{\sqrt{\rho_0}}\bar{\beta})\partial_{x_k} W}{\int_0^T\e^{-\i h (T-s)}P^{(k)}_s\partial_{x_k}(h)^{-1} W \d s}\nonumber\\
&+ R(P,T)_k\,.
\end{align*}
Since $T>0$ is arbitrary we have
\begin{align}\label{eq:mainequation}
\dot{P}_t=L(P)(t)+R(P,t)\,,
\end{align}
where $L(P)$ is defined as
\begin{align*}
L(P):=\begin{pmatrix}L(P^{(1)})\\L(P^{(2)})\\L(P^{(3)})\end{pmatrix}\,.
\end{align*}
\textbf{Remark:} From now on, we will write $t$ for $T$ for esthetic reasons.\\

\section{The existence of the solution in the infinite time interval}
It is hard to derive a decay estimate for $P_t$ directly from (\ref{eq:mainequation}). In what follows we will rearrange terms until a fixed point theorem becomes applicable.

We will express the solution of the full equation (\ref{eq:mainequation}) in terms of the solution $K(t)$ of one part of the linear equation,
\begin{align}\label{eq:k}
\dot{K}(t)&=Z\Re\scalar{[1-g(h)^{-1}W]\partial_{x_1} W}{\int_0^t\e^{-\i h (t-s)}K(s)\partial_{x_1}(h)^{-1} W\d s}\,,\\
K(0)&=1.
\end{align}
Here the constant $Z\in \mathbb{R}^{+}$ is defined as $$Z:=2\frac{\rho_0}{M}.$$ In Appendix \ref{appendix:k} we prove the following lemma,
\begin{lem}\label{lem:k}
Let $K(t)$ be a solution of equation (\ref{eq:k}) with $K(0)=1$. Then there exist real constants $C_1,\ C_2$ such that as $t\to\infty$
  \begin{align}\label{eq:k-estimate}
    ZK(t)=\frac{3}{\sqrt{2}}\pi^{-\frac{5}{2}}(1+C_1g)t^\mez+C_2 t^{-1}+O(t^{-\frac{3}{2}})\,.
  \end{align}
\end{lem}
With $K(t)$ at hand, we can write the Duhamel-like formula
\begin{align}\label{eq:temPt}
P_t=K(t)P_0+Z\int_0^tK(t-s)\Re\scalar{[1-g(h)^{-1}W]\partial_{x_1}W}{\e^{-\i h s}\partial_{x_1}(h)^{-1}W}\int_0^sP_{s_1}\d s+\int_0^tK(t-s)R(P,s)\d s\,.
\end{align}
We now manipulate ~\eqref{eq:temPt} and ~\eqref{eq:mainequation} to obtain an effective equation for $P_{t}.$
Since the procedure is very similar to \cite{Froehlich102} we will go through the steps quickly.

We integrate both sides of ~\eqref{eq:mainequation} from $0$ to $t$, then multiply by $K(t)$ to obtain
\begin{align*}
K(t)P_{t}=K(t)P_0+K(t)\int_{0}^{t}\Phi(s) \d s
\end{align*} where $\Phi(s)$ stands for various terms on the right hand side of ~\eqref{eq:mainequation}.
Now we use this equation to subtract ~\eqref{eq:temPt}, then manipulate the linear terms of $P_{t}$ in ~\eqref{eq:temPt} and $\Phi$ and use the observation $\Re\langle[1-gh^{-1}W]\partial_{x_1}W,\ (ih)^{-1}\partial_{x_1}h^{-1}W \rangle=0$
to find
\begin{align}\label{eq:contraction1}
P_t=\frac{1}{1-K(t)}A(P)(t)+\frac{1}{1-K(t)}\int_0^t[K(t-s)-K(t)]R(P,s)\d s\,,
\end{align}
where the linear operator $A$ is defined by
\begin{align}\label{eq:defina}
 A(P)(t)=&-Z\int_0^t[K(t-s)-K(t)]\Re\scalar{[1-g(h)^{-1}W]\partial_{x_1} W}{\e^{-\i h s}\partial_{x_1}(h)^{-1}W}\int_s^tP_{s_1}\d s_1\d s\nonumber\\
&+ZK(t)\Re\scalar{[1-g(h)^{-1}W]\partial_{x_1}W}{(-\i h)^{-1}\int_0^t[\e^{-\i h(t-s)}-\e^{-\i h t}]P_s\d s \partial_{x_1}(h)^{-1}W}\\
&+Z\int_0^tK(t-s)\Re\scalar{[1-g(h)^{-1}W]\partial_{x_1}W}{\e^{-\i h s}\partial_{x_1}(h)^{-1}W}\d s\int_0^tP_{s_1}\d s_1\,.\nonumber
\end{align}

In the rest of the paper we focus on studying ~\eqref{eq:contraction1}. We start with casting the equation in a Banach space setting, so that a fixed point theorem applies. In order to rewrite the equation for $P_t$ as the integral equation (\ref{eq:contraction1}) we had to divide by $1-K(t)$, which needs some care for small values of $t$ since $K(t)\to 1$ as $t\to 0$. But because we know from Lemma \ref{lem:k} that $K(t)\to 0$ for $t\to\infty$, it suffices to wait long enough before dividing by $1-K(t)$. Therefore, we divide the time interval $[0,\infty)$ into two parts $[0,T_{\rm loc})$ and $[T_{\rm loc},\infty)$. Introduce a family of Banach spaces that reflects the self-consistent assumption $P_t=O(t^{\mez-\delta})$,
\begin{align*}
B_{\delta,T_{\rm loc}}:=\{f:t^{\pez+\delta} f\in L^\infty[T_{\rm loc},\infty)\}
\end{align*}
with norm
\begin{align*}
\norm{f}_{\delta,T_{\rm loc}}:=\norm{t^{\pez+\delta}f}_\infty\,.
\end{align*}
On the finite interval $[0,T_{\rm loc})$ we can use standard existence and uniqueness results to solve (\ref{eq:temPt}), and for the infinite interval $[T_{\rm loc},\infty)$ we use a fixed point theorem. Introduce the Heaviside function $\chi_{T_{\rm loc}}:=\11_{[0,T_{\rm loc})}$ and rewrite (\ref{eq:contraction1}) as
\begin{align}\label{eq:contraction2}
P_t=\Upsilon((1-\chi_{T_{\rm loc}})P)(t)+G_t,
\end{align}
where
\begin{align*}
\Upsilon((1-\chi_{T_{\rm loc}})P)(t)&:=\frac{1}{1-K(t)}A((1-\chi_{T_{\rm loc}})P)(t)+\frac{1}{1-K(t)}\int_0^t[K(t-s)-K(t)][R(P,s)-R(\chi_{T_{\rm loc}}P,s)]\d s\\
G_t&:=\frac{1}{1-K(t)}A(\chi_{T_{\rm loc}}P)(t)+\frac{1}{1-K(t)}\int_0^t[K(t-s)-K(t)]R(\chi_{T_{\rm loc}}P,s)\d s\,.
\end{align*}

Now we present the strategy of applying the fixed point theorem. To this end two criteria have to be verified: the nonlinear operator $\Upsilon$ maps a small neighborhood of $0$, in the space $B_{\delta,T},$ into itself and is contractive; the function $G_{t}$ is sufficiently small in the space $B_{\delta,T}$.

The following two propositions, to be proven in the appendix, show that for $T_{\rm loc}$ large enough, $\Upsilon((1-\chi_{T_{\rm loc}})P)(t):B_{\delta,T_{\rm loc}}\to B_{\delta,T_{\rm loc}}$ is indeed a contraction, and $G_t$ is small in $B_{\delta,T_{\rm loc}}$ if the initial conditions for $P$ and $\beta$ are small enough, which will allow us to prove the main theorem. Recall the definition of $\Omega$ from ~\eqref{eq:defOmega}.
\begin{prop}\label{prop:contraction-lemma}There is an $M>0$ such that for $T_{\rm loc}\geq M$ the mapping $\Upsilon((1-\chi_{T_{\rm loc}})P)(t):B_{\delta,T_{\rm loc}}\to B_{\delta,T_{\rm loc}}$ is a contraction, or more precisely:
  \begin{itemize}
  \item[(1)] For any function $q\in B_{\delta,T_{\rm loc}}$
    \begin{align*}
     t^{\pez+\delta} \big|\frac{1}{1-K(t)}A((1-\chi_{T_{\rm loc}})q_t)\big|\leq [\frac{1}{\pi}\Omega(\delta)+\eps(T_{\rm loc})+O(g)]\norm{q_t}_{\delta,T_{\rm loc}}\,,
    \end{align*}
where $\eps(T_{\rm loc})\to 0$ as $T_{\rm loc}\to\infty$.
\item[(2)]Recall that the solution $P$ exists in the time interval $[0,T_{\rm loc}]$ according to Theorem \ref{thm:localwp}. Suppose that $Q_1,Q_2:[0,\infty)\to\RR^3$ are two functions satisfying
  \begin{align*}
    Q_1(t)=Q_2(t)=P_t \quad \textrm{for $t\in[0,T_{\rm loc}]$}\,,
  \end{align*}
and in the interval $[T_{\rm loc},\infty)$
\begin{align*}
  \norm{Q_1}_{\delta,T_{\rm loc}},\norm{Q_2}_{\delta,T_{\rm loc}}\ll 1\,.
\end{align*}
Then,
  \end{itemize}
  \begin{align*}
    t^{\pez+\delta}\big|\frac{1}{1-K(t)}\int_0^t[K(s-t)-K(t)][R(Q_1,s)-R(Q_2,s)]\d s \big|\lesssim \(\norm{Q_1}_{\delta,T_{\rm loc}}+\norm{Q_2}_{\delta,T_{\rm loc}}\)\norm{Q_1-Q_2}_{\delta,T_{\rm loc}}\,.
  \end{align*}
\end{prop}
\begin{prop}\label{prop:contraction-lemma2}Suppose that $T_{\rm loc}\geq M$ (see Proposition \ref{prop:contraction-lemma}) and $|P_0|,\norm{\x^3\beta_0}_2\leq\eps_0(T_{\rm loc})$ (see Theorem \ref{thm:localwp}).  Then $G_t$ is in the Banach space $B_{\delta,T_{\rm loc}}$, and its norm is small. Specifically, for any $t\geq T_{\rm loc}$
$$\begin{array}{ccc}
t^{\pez+\delta}\big|\frac{1}{1-K(t)}A(\chi_{T_{\rm loc}}P)(t)\big|&\leq&\eps(T_{\rm loc})\\
t^{\pez+\delta}\big|\frac{1}{1-K(t)}\int_0^t[K(t-s)-K(t)]R(\chi_{T_{\rm loc}}P,s)\d s\big|&\leq&\eps(T_{\rm loc})\,,
\end{array}$$
with $\eps(T_{\rm loc})\to 0$ as $T_{\rm loc}\to\infty$.
\end{prop}

\paragraph{Key ideas} As we have stated above, the proof of these propositions can be found in the appendix, but we want to give here the key ideas.

To prove that $G_t=G(\chi_{T_{\rm loc}}P_{\cdot})(t)$ is small we need to choose the initial conditions suitably small. For notice that $G_t$ is defined in terms of $\chi_{T_{\rm loc}}P_{\cdot}$ and the initial condition $\beta_0$; moreover, Theorem \ref{thm:localwp} states that for $|P_0|$ and $\norm{\x^3\beta_0}_2$ small enough we have $|P_t|\leq T_{\rm loc}^{-2}$ for any $0\leq t\leq T_{\rm loc}$. This makes it plausible that we can prove $\norm{G_t}_{\delta, T_{\rm loc}}\to 0$ as $T_{\rm loc}\to \infty$.

The proof of Proposition \ref{prop:contraction-lemma} is more involved because $\Upsilon_t=\Upsilon((1-\chi_{T_{\rm loc}})P_{\cdot})(t)$ is defined in terms of the infinite time trajectory, $\chi_{[T_{\rm loc},\infty)}P_{\cdot}$. For brevity, write
\begin{align*}
  f(s):=\Re\scalar{[1-g(h)^{-1}W]\partial_{x_1} W}{\e^{-\i h s}\partial_{x_1}(h)^{-1}W}\,,
\end{align*}
so that the first term in the definition of $A(P)(t)$ takes the form
\begin{align*}
\Gamma_1((1-\chi_{T_{\rm loc}})P_\cdot):=-Z\int_0^t[K(t-s)-K(t)]f(s)\int_s^t(1-\chi_{T_{\rm loc}})P_{s_1}\d s_1\d s\,.
\end{align*}
We already know the decay of $K(t)=O(t^\mez)$ from Lemma \ref{lem:k}, and standard dispersive estimates give
\begin{align*}
f(t)=C t^\mdz(1+\widetilde{C}g)+o(t^{-\frac{3}{2}})\,,
\end{align*}
for some explicit constant $C$. This, combined with the self-consistent assumption $P_t=O(t^{\mez-\delta})$, is enough to prove
\begin{align*}
  \norm{\Gamma_1((1-\chi_{T_{\rm loc}})P_\cdot)}_{\delta,T_{\rm loc}}\leq [\Omega_1(\delta)+\epsilon(T_{loc})]\norm{P}_{\delta,T_{\rm loc}} \,,
\end{align*}
where the constant $\Omega_1(\delta)$ is defined as $\Omega_1(\delta):=\frac{1}{(1-2\delta)\pi}\int_0^1\frac{1}{1+(1-r)^\pez}(1-r)^\mez[r^\mez-r^{-\delta}]\d r$, and $\epsilon(T_{loc})\rightarrow 0$ as $T_{loc}\rightarrow \infty.$
A similar computation can be made for the second term in the definition of $A(P)(t)$ with $\Omega_1(\delta)$ replaced by $\Omega_2(\delta):=\frac{1}{\pi}\int_0^1\frac{1}{1+(1-r)^\pez}(1-r)^\mez r^{\pez-\delta}\d r\,.$ In the end, $\delta$ has to be chosen such that $\Omega_1(\delta)+\Omega_2(\delta)<1$ in order to effect a contraction.

The third term in the definition of $A(P)(t)$ can be rewritten as
\begin{align*}
\Gamma_3((1-\chi_{T_{\rm loc}})P_{\cdot}):&=Z\int_0^tK(t-s)f(s)\d s\int_0^t(1-\chi_{T_{\rm loc}})P_{s_1}\d s_1\\&=Z\int_0^tK(s)f(t-s)\d s\int_0^t(1-\chi_{T_{\rm loc}})P_{s_1}\d s_1\\
&=\dot{K}(t)\int_0^t(1-\chi_{T_{\rm loc}})P_{s_1}\d s_1\,,
\end{align*}
and by a modification of the proof of Lemma \ref{lem:k} we prove that $\dot{K}(t)=O(t^\mdz)$, so that it is again straight forward to establish
\begin{align*}
  \norm{\Gamma_3((1-\chi_{T_{\rm loc}})P_\cdot)}_{\delta,T_{\rm loc}}\leq \eps(T_{\rm loc})\norm{P}_{\delta,T_{\rm loc}} \,,
\end{align*}
where $\eps(T_{\rm loc})\to 0$ as $T_{\rm loc}\to\infty$.

The proof of point (2) of Proposition \ref{prop:contraction-lemma} involves lengthy computations the core of which are the propagator estimates proved in Appendix \ref{sec:propagator-estimates}.

\section{Proof of Main Theorem ~\ref{thm:main}}\label{mainproof}
As discussed before, we divide the time interval $[0,\infty)$ into two parts, $[0,T_{\rm loc})$ and $[T_{\rm loc},\infty)$ with $T_{\rm loc}$ being a large constant. The existence of the solution in the finite domain was proven in Theorem \ref{thm:localwp}. For the infinite domain, Propositions \ref{prop:contraction-lemma} and \ref{prop:contraction-lemma2} enable us to apply the contraction lemma on (\ref{eq:contraction2}) to see that there exists a small solution $P$ in the space $B_{\delta,T_{\rm loc}}$. By the definition of $B_{\delta,T_{\rm loc}}$ we have proven (\ref{momentumdecay}).

To prove \eqref{deltadecay} it is sufficient to prove that $$\|\langle x\rangle^{-3}\delta_{T,T}\|_{2}\lesssim T^{-\frac{1}{2}}\ \text{for any}\ T\geq 0$$ where the function $\delta_{t}=\delta_{t,T}$ is defined in ~\eqref{eq:decomp}. This has been proved in Proposition ~\ref{prop:delta}.

The existence of $X_{\infty}$ is resulted by the fact $P_{\cdot}=O(t^{-\frac{1}{2}-\delta})$ is integrable in the region $[T_{\rm loc},\infty).$

The proof is complete.  \hfill $\square$

\appendix
\section{Proof of Proposition ~\ref{prop:delta}}\label{sec:ControlRem}
For any time $\tau\leq T$ we apply Duhamel's principle to rewrite (\ref{eq:delta}) as
  \begin{align}\label{eq:deltaduha}
  &\delta_\tau&&=&&\e^{-\i\h\tau}\delta_0-\i g\int_0^\tau\e^{-\i\h(\tau-s)}[W^{X_s}-\W]\delta_s\d s\nonumber\\
&&&&&-\frac{\sqrt{\rho_0}}{M}\sum_{|\alpha|=1}^{N_0}\frac{1}{\alpha!}\int_0^\tau\e^{-\i\h (\tau-s)}\partial_x^\alpha(\h)^{-1}\W P_s\alpha(X_s-X_T)^{\alpha-1}\d s
+\i\int_0^\tau\e^{-\i\h (\tau-s)}G_1(s)\d s\nonumber\\
&&&=:&&\sum_{n=1}^4B_n\,.
  \end{align}
Now we estimate each term on the right hand side of \eqref{eq:deltaduha}. Recall the definition of $\mu(T)$ in (\ref{majorant}) and the assumption $\mu(T)\leq 1$. By the definition of $\delta_0$ and the propagator estimates of Proposition \ref{prop:propagator} we have
\begin{align}
&\norm{\x^{-3}B_1}_2&&=&&\norm{\x^{-3}\e^{-\i\h \tau}[\beta_0-\bar{\beta}^{X_T}-\sqrt{\rho_0}\sum_{|\alpha|=1}^{N_0}\frac{1}{\alpha!}(X_0-X_T)^\alpha\partial_x^\alpha(\h)^{-1}\W]}_2\label{eq:b111}\\
&&&\leq &&\norm{\x^{-3}\e^{-\i\h \tau}\beta_0}_2+\norm{\x^{-3}\e^{-\i\h \tau}\bar{\beta}^{X_T}}_2\nonumber\\
&&&&&+\sqrt{\rho_0}\sum_{|\alpha|=1}^{N_0}\frac{1}{\alpha!}|X_0-X_T|^\alpha\norm{\x^{-3}\e^{-\i\h \tau}\partial_x^\alpha(\h)^{-1}\W}_2\nonumber\\
&&&\lesssim&&(1+\tau)^\mdz\norm{\x^3\beta_0}_2+(1+\tau)^\mez+(1+\tau)^\mdz\mu(T)\nonumber
\end{align}
where in the third line we used the fact
\begin{align*}
|X_0-X_T|\leq \int_0^T|P_s|\d s\lesssim \mu(T)\,.
\end{align*}
For the last line we recall the overarching hypothesis of Theorem \ref{thm:main} $\norm{\x^3\beta_0}_2\leq \eps_0$.\\
For $B_3$ we have
\begin{align*}
\norm{\x^{-3}B_3}_2\lesssim&\mu(T)\int_0^\tau(1+\tau-s)^\mdz(1+s)^\mfv\d s\\
\lesssim &\mu(T)(1+\tau)^\mfv\,;
\end{align*}
recall that we only consider $\delta\in(\frac{1}{2},\delta^*)$ and $\delta^*<1$. Similarly for $B_4$,
\begin{align*}
\norm{\x^{-3}B_4}_2\lesssim&\mu^{N_0+1}(T)\int_0^\tau(1+\tau-s)^\mdz(1+s)^\mdz\d s\\
\lesssim &\mu^{N_0+1}(T)(1+\tau)^\mdz\,.
\end{align*}
Since $B_2$ depends on $\delta_\tau$, we have to proceed in a different way. Define the function $Q$ by
\begin{align*}
Q(\tau):=\max_{0\leq s\leq \tau\leq T}(1+s)^\pez\norm{\x^{-3}\delta_s}_2\,.
\end{align*}
Then $B_2$ admits the estimate
\begin{align*}
  \norm{\x^{-3}B_2}_2\lesssim&g\int_0^\tau(1+\tau-s)^\mdz\norm{\x^3(\W-W^{X_s})\delta_s}_2\d s\\
  \lesssim &gQ(\tau)\int_0^\tau(1+\tau-s)^\mdz |X_t-X_s|\ (1+s)^\mez\d s\\
  \lesssim & gQ(\tau)\mu(\tau)\int_{0}^{\tau}(1+\tau-s)^\mdz\ [(1+s)^{\frac{1}{2}-\delta}-(1+\tau)^{\frac{1}{2}-\delta}]\ (1+s)^\mez\d s\\
  \lesssim & gQ(\tau)\mu(\tau) (1+\tau)^{-\frac{1}{2}-\delta}
\end{align*}
In the first line, we used the fact
\begin{align*}
|\x^3W^{X_\cdot}|\lesssim\x^{-3}\,,
\end{align*}
which holds since $|X_\cdot|$ is bounded, and in the last step Lemma ~\ref{LM:convo} has been used.

Collecting the estimates above we find
\begin{align*}
(1+\tau)^\pez\norm{\x^{-3}\delta_\tau}\lesssim gQ(\tau)+1+\eps_0+\mu(T)\,,
\end{align*}
which by definition of $Q(\tau)$ yields for any $0\leq \tau\leq T$
\begin{align*}
  Q(\tau)\lesssim gQ(\tau)+1+\eps_0+\mu(T)\,.
\end{align*}
As $g$ is small we obtain
\begin{align*}
   Q(\tau)\lesssim 1+\eps_0+\mu(T)\lesssim 1\,,
\end{align*}
which is the desired estimate.

The proof is complete.
\begin{flushright}
$\square$
\end{flushright}

In the proof of ~\eqref{eq:mira} the following result has been used.
\begin{lem}\label{LM:convo}
\begin{align}
\int_{0}^{t}(1+t-s)^{-\frac{3}{2}} (s^{\frac{1}{2}-\delta}-t^{\frac{1}{2}-\delta})(1+s)^\mez\d s\lesssim & (1+t)^{-\frac{1}{2}-\delta}.
\end{align}
\end{lem}
\begin{proof}
We start with deriving a convenient form
\begin{align*}
(t^{\pez-\delta}-s^{\pez-\delta})=-t^{\pez-\delta}s^{\pez-\delta}(t^{\delta-\pez}-s^{\delta-\pez})\,.
\end{align*} To estimate the term $t^{\delta-\pez}-s^{\delta-\pez}$ we consider two different regimes, $0\leq s\leq \frac{t}{2}$ and $\frac{t}{2}\leq s\leq t$. In the first regime we use direct estimate, for the second we use Taylor expansion to find that for $s\leq t$ and any $\eps>0$ there exists a constant $c(\epsilon)>0$
\begin{align*}
t^\eps-s^\eps\leq c(\epsilon)\frac{t-s}{t^{1-\eps}}
\end{align*}
This implies
\begin{align*}
\int_{0}^{t}(1+t-s)^{-\frac{3}{2}} (s^{\frac{1}{2}-\delta}-t^{\frac{1}{2}-\delta})(1+s)^\mez\d s\lesssim &
\int_0^t(1+t-s)^{-\frac{1}{2}} t^{-1}s^{\pez-\delta}(1+s)^\mez\d s\\
\leq & t^{-1}\int_0^t(t-s)^{-\frac{1}{2}} s^{-\delta}\d s\\
\leq &t^{-\frac{1}{2}-\delta} \int_{0}^{1}(1-s)^{-\frac{1}{2}} s^{-\delta}\ ds\\
\lesssim & t^{-\frac{1}{2}-\delta}
\end{align*} where in the second step we rescaled variable $s\rightarrow t s$ and in the last step we used $\delta<1.$

To remove the singularity at $t=0$ we use a direct estimate on the expression to prove it is bound for $t\leq 1.$

The proof is complete.
\end{proof}

The following result will be used later.
Define a new function $\phi$ by
\begin{align}\label{eq:difPhi}
\phi_t:=\delta_t+\e^{\i h^{X_{T}}t}\bar{\beta}^{X_{T}}.
\end{align}
\begin{cor}\label{cor:oneterm}
\begin{align*}
\|\langle x\rangle^{-3}\phi_t\|_2\lesssim (1+t)^{-\frac{1}{2}-\delta}
\end{align*}
\end{cor}
\begin{proof}
The proof is based on an improvement of the proof of Proposition ~\eqref{eq:delta}. Observe that the only term not of order $t^{-\frac{1}{2}-\delta}$ (or smaller) is $-\e^{\i h^{X_{T}}t}\bar{\beta}^{X_{T}}$, see ~\eqref{eq:b111}. Recall that $\delta<1.$ Hence by removing this term we obtain the desired estimate.

The proof is complete.
\end{proof}
\section{Proof of Lemma \ref{lem:k}}\label{appendix:k}
We follow the strategy of \cite{Froehlich102}. Define $Z:=\frac{2\rho_0}{M}$ and a function $G:\RR\to\CC$ by
\begin{align}\label{eq:A1}
G(k+\i0):=&\frac{\i}{2}\scalar{(h+k+\i0)^{-1}\partial_{x_1}(h)^{-1}W}{[1-g(h)^{-1}W]\partial_{x_1} W}\nonumber\\
-&\frac{\i}{2}\scalar{[1-g(h)^{-1}W]\partial_{x_1}W}{(h-k-\i0)^{-1}\partial_{x_1}(h)^{-1} W}
\end{align}
Now, we relate $G$ to the solution $K$:
\begin{lem}\label{lem111}The solution $K$ of (\ref{eq:k}) takes the form
  \begin{align*}
   K(t)&=-\frac{1}{\pi}\int_{-\infty}^\infty\Re \frac{1}{\i k+ZG(k+\i0)}\cos kt\d k.
  \end{align*}
In particular,
\begin{align*}
  K(t)=0 \qquad \textrm{for $t<0$}.
\end{align*}
\end{lem}
The proof of Lemma \ref{lem111} is done as in \cite{Froehlich102} and is not repeated here. With this explicit expression for $K$, we can prove Lemma \ref{lem:k} with the help of the following lemma,
\begin{lem}\label{lem:G}The function $G(k+\i0)$ satisfies
  \begin{align*}
  G(k+\i0)=\begin{cases}(\i-1)\frac{\pi^2}{3}(1+O(g))k^\pez+C_1k+O(k^\pdz)& \textrm{if $k>0$}\\
(-\i-1)\frac{\pi^2}{3}(1+O(g))|k|^\pez+C_2 k+O(|k|^\pdz)& \textrm{if $k<0$}\end{cases},
  \end{align*}
with $C_1,\ C_2$ being some constants.
\end{lem}
Lemma \ref{lem:G} is proven at the end of this section.
\begin{proof}[Proof of Lemma \ref{lem:k}]Decompose $K(t)$ into two parts,
  \begin{align*}
    K(t)=K_+(t)+K_-(t)\,,
  \end{align*}
with
\begin{align*}
  K_+(t):=-\frac{1}{\pi}\int_0^\infty\Re \frac{1}{\i k+ZG(k+\i0)}\cos kt\d k
\end{align*}
and
\begin{align*}
  K_-(t):=-\frac{1}{\pi}\int_{-\infty}^0\Re \frac{1}{\i k+ZG(k+\i0)}\cos kt\d k
\end{align*}
Define a new function $g:\RR^+\to\RR$ by
\begin{align*}
|k|^\mez g(|k|^\pez):&=-\frac{1}{\pi}\Re \frac{1}{\i k+ZG(k+\i0)}\\
&=-\frac{1}{Z\pi} \frac{\Re G}{(\frac{k}{Z}+\Im G)^2+(\Re G)^2}\\
&=\frac{3(1+O(g))}{2\pi^3Z} |k|^\mez(1+O(k^\pez))\,,
\end{align*}
where we used the explicit form of $G(k+\i0)$ of Lemma \ref{lem:G}. By construction, the function $g$ is smooth on $[0,\infty)$ and satisfies (because $G(k)$ is bounded as $k\to\infty$)
\begin{align*}
|g(\rho)|\leq C(1+\rho)^{-3}\,.
\end{align*}
Compute directly to obtain
\begin{align*}
K_+(t)&=\int_0^\infty k^\mez g(k^\pez)\cos kt\d k\\
&=2\int_0^\infty g(\rho)\cos(\rho^2 t)\d\rho\\
&=2g(0)\int_0^\infty \cos(\rho^2 t)+D
\end{align*}
with $D$ defined as
\begin{align*}
D:=2\int_0^\infty[g(\rho)-g(0)]\cos(\rho^2 t)\d\rho\,.
\end{align*}
The first term on the right hand side is the dominant one:
\begin{align*}
  2g(0)\int_0^\infty\cos(\rho^2 t)\d\rho=2g(0)t^\mez\int_0^\infty \cos x^2\d x = \frac{3(1+O(g))}{2\sqrt{2}Z}\pi^{-\frac{5}{2}} t^\mez\,,
\end{align*}
where we used the Fresnel integral $\int_0^\infty \cos x^2\d x=(\pi/8)^\pez$.

We prove now that $D$ is a correction of order $t^\mdz$. This implies
\begin{align*}
K_+=\frac{3(1+O(g))}{2\sqrt{2}Z}\pi^{-\frac{5}{2}} t^\mez+O(t^{-1})\,.
\end{align*}
Since we find by completely analogous computation
\begin{align*}
K_-=\frac{3(1+O(g))}{2\sqrt{2}Z}\pi^{-\frac{5}{2}} t^\mez+O(t^{-1})
\end{align*}
the claim follows.

To estimate $D$ we first integrate by parts:
\begin{align*}
|D|&=t^{-1}|\int_0^\infty \rho^{-1}[g(\rho)-g(0)]\partial_\rho\sin(\rho^2t)\d\rho|\\
&=t^{-1}|\int_0^\infty H(\rho)\sin(\rho^2t)\d\rho|
\end{align*}
with $H(\rho):=\partial_\rho(\rho^{-1}[g(\rho)-g(0)])$ a smooth function satisfying $|H(\rho)|\lesssim (1+\rho)^{-2}$. Write $H(\rho)=H(0)+\rho[\rho^{-1}(H(\rho)-H(0))]$ and perform again integration by parts to obtain
\begin{align*}
|D|=t^{-1}|H(0)||\int_0^\infty\sin(\rho^2t)\d\rho|+\frac{1}{2}t^{-2}\lim_{\rho\to0}\frac{|H(\rho)-H(0)|}{\rho}+\frac{1}{2}t^{-2}|\int_0^\infty\partial_\rho[\rho^{-1}(H(\rho)-H(0))]|\d\rho\,.
\end{align*}
The first term on the right hand side can be computed explicitely,
\begin{align*}
t^{-1}|H(0)||\int_0^\infty\sin(\rho^2t)\d\rho|=t^\mdz|H(0)|\sqrt{\frac{\pi}{8}}\,,
\end{align*}
and the second term is obviously of order $t^{-2}$. The last term is controlled by
\begin{align*}
t^{-2}\int_0^\infty(1+\rho)^{-2}\d\rho\lesssim t^{-2}
\end{align*}
by the fact that $|\partial_\rho[\rho^{-1}(H(\rho)-H(0))]|\lesssim (1+\rho)^{-2}$.
\end{proof}

\begin{proof}[Proof of Lemma \ref{lem:G}]

The basic idea is to expand $(h+k)^{-1}$ around $h^{-1}$. By classical results, see e.g.\ \cite{jensen79}, if the constant $|g|$ in $h=-\Delta+gW$ is sufficiently small and $W$ decays sufficiently fast at $\infty,$ then $h$ has no zero-resonance or eigenvectors. This together with the discussions above and results in \cite{jensen79} implies that
  \begin{align}\label{eq:expansion}
    (h+k)^{-1}=B_0+\zeta B_1+\zeta^2 B_2+O(\zeta^3)\,,
  \end{align}
in the topology of $\mathcal{B}(L^{2,3},L^{2,-3})$, $B_i$ being operators in $\mathcal{B}(L^{2,3},L^{2,-3})$, namely
\begin{align*}
B_0&=(1+(-\Delta)^{-1}gW)^{-1}(-\Delta)^{-1}\\
B_1&=\frac{1}{4\pi}\scalar{\cdot}{(1+(-\Delta)^{-1}gW)1}(1+(-\Delta)^{-1}gW)1\,
\end{align*} and the variable $\zeta$ is defined by $\zeta:=k^{\frac{1}{2}},$ where $k$ is in the domain $\mathbb{C}\backslash\mathbb{R}^+,$ and $k^{\frac{1}{2}}=k^{\frac{1}{2}}>0$ for $k>0.$

A minor difficulty in the present situation is that we cannot apply the expansion ~\eqref{eq:expansion} directly because $\partial_{x_1}h^{-1}W\notin L^{2,3}$. To make ~\eqref{eq:expansion} still applicable we observe that for any $k\not=0$, $g\in \mathbb{R}$
\begin{align}\label{eq:simpleMan}
(-\Delta+gW)^{-1}(-\Delta+gW+k\pm 0i)^{-1}=\frac{1}{k}[(-\Delta+gW)^{-1}-(-\Delta+gW+k\pm 0i)^{-1}].
\end{align} There is one more minor difficulty: In applying this equation to study $G(k+i0)$ it might happen that $G(k+i0)$ is singular at $k=0.$ For this we use the presence $\partial_{x}$ and the fact that $W$ is spherically symmetric, or simply using symmetries, to prove that the singular terms in $G(k+i0)$ are identically zero.

In the next we carry out the ideas presented above.

Observe that $G(k+i0)$ contains two term, denote the two terms by $G_1(k+i0)$ and $G_2(k+i0)$, i.e., $$G(k+i0)=G_{1}(k+i0)-G_2(k+i0).$$ We start with studying $G_{1}(k+i0).$
Define a function $W_1$ by
\begin{align*}
\partial_{x_1}W_1:=[1-g(h)^{-1}W]\partial_{x_1} W\,.
\end{align*}
Clearly, $W_1$ is spherically symmetric and rapidly decaying. Use the second resolvent identity to rewrite the first term of $G(k)$,
\begin{align}
G_1(k+i0)=&\frac{\i}{2}\scalar{[(-\Delta +  k+\i0)^{-1}-(h +  k+\i0)^{-1}gW(-\Delta +  k+\i0)^{-1}]\partial_{x_1}[(-\Delta)^{-1}-(-\Delta)^{-1}gWh^{-1}]W}{\partial_{x_1}W_1}\nonumber\\
=&\frac{\i}{2}\scalar{(-\Delta +  k+\i0)^{-1}\partial_{x_1}(-\Delta)^{-1}W}{\partial_{x_1}W_1}\nonumber\\
&-\frac{\i}{2}\scalar{(-\Delta +  k+\i0)^{-1}\partial_{x_1}(-\Delta)^{-1}gWh^{-1}W}{\partial_{x_1}W_1}\nonumber\\
&+\frac{\i}{2}\scalar{(h +  k+\i0)^{-1}gW(-\Delta +  k+\i0)^{-1}\partial_{x_1}(-\Delta)^{-1}W}{\partial_{x_1}W_1}\nonumber\\
&-\frac{\i}{2}\scalar{(h +  k+\i0)^{-1}gW(-\Delta +  k+\i0)^{-1}\partial_{x_1}(-\Delta)^{-1}gWh^{-1}W}{\partial_{x_1}W_1}\nonumber\\
=&A_1+A_2+A_3+A_4
\end{align} where the terms $A_l,\ l=1,2,3,4,$ are naturally defined as the four terms on the right hand side.

We claim that there exist constants $C_{l,n},\ l=1,2,3,4\ n=1,2,3$ such that
\begin{align}\label{eq:claimTaylor}
A_{l}=C_{l,1}+C_{l,2}k^{\frac{1}{2}}+C_{l,3}k+O(k^{\frac{3}{2}})
\end{align} and moreover the constants $C_{l,n},\ l\not=1$ are of order $O(g)$.

Suppose the claim holds. Then by the same techniques we prove the second term $G_2(k+i0)$ in $G(k+i0)$ can also be expanded in a form similar to ~\eqref{eq:claimTaylor}. It is not difficult to see that the constant terms cancel each other if $G(k+\i0)$ is defined for the point $k=0$ (which we will prove), which makes $G(0+i0)=0.$ For the term of order $k^{\frac{1}{2}}$ we will compute the coefficient in the proof of the claim.

We start with proving ~\eqref{eq:claimTaylor} for $l=1.$
$B_1$ is the main term in the sense that only $B_1$ does not contain the small factor $g$. We rewrite it as
\begin{align*}
\frac{\i}{2}\scalar{(-\Delta +  k+\i0)^{-1}\partial_{x_1}(-\Delta)^{-1}W}{\partial_{x_1}W_1}=\frac{\i}{6}\scalar{(-\Delta +  k+\i0)^{-1}W}{W_1}\,,
\end{align*}
for which (\ref{eq:expansion}) becomes applicable. The constant term in the expansion vanishes in the difference (\ref{eq:A1}).  For the $k^{1/2}$-term we get (consider first $k>0$)
\begin{align*}
&\frac{1}{24\pi}k^{1/2}(\i-1)\;\scalar{W}{(1+(-\Delta)^{-1}gW}\scalar{(1+(-\Delta)^{-1}gW}{W_1}\\
=&\frac{1}{24\pi}k^{1/2}(\i-1)\big[\scalar{W}{1}\scalar{1}{W}+O(g)\big]\,,
\end{align*}
where we used $W_1=W+O(g)$ in $\norm{\cdot}_\infty$. Using $\scalar{W}{1}=\scalar{1}{W}=(2\pi)^\pdz\widehat{W}(0)$ the last line equals
\begin{align*}
\frac{\pi^2}{3}k^{1/2}(\i-1)(1+O(g))  \,.
\end{align*}

The case $l=2$ is treated analogously and gives a contribution of order $k^\pez O(g)$.

For the case $l=3$ we rewrite the expression as
\begin{align}\label{eq:a31}
A_3=&\frac{\i}{2}\scalar{(h +  k+\i0)^{-1}gW(-\Delta +  k+\i0)^{-1}\partial_{x_1}(-\Delta)^{-1}W}{\partial_{x_1}W_1}\nonumber\\
=&g\frac{\i}{2}\scalar{(-\Delta +  k+\i0)^{-1}\partial_{x_1}(-\Delta)^{-1}W}{W(h +  k-\i0)^{-1}\partial_{x_1}W_1}\,.
\end{align}
Now we apply ~\eqref{eq:simpleMan} (with $g=0$) and ~\eqref{eq:expansion} above to expand the expression in $k^{\frac{1}{2}}.$ Observe that the singular terms $k^{-1}$ and $k^{-\frac{1}{2}}$ vanish by symmetry. Hence we have proved ~\eqref{eq:claimTaylor}

The term $A_4$ can be treated in the exact same way and yields a contribution of order $k^\pez O(g^2)$.
\end{proof}

\section{Proof of Point (1) of Proposition \ref{prop:contraction-lemma}}
Similar to \cite{Froehlich102} we decompose the linear operator $A$, defined in (\ref{eq:defina}), as follows,
\begin{align*}
  A((1-\chi_{T_{\rm loc}})q):=\sum_{k=1}^3\Gamma_k((1-\chi_{T_{\rm loc}})q)\,,
\end{align*}
where the terms $\Gamma_k$ are defined as
\begin{align*}
\Gamma_1&:=-Z\int_0^t[K(t-s)-K(t)]\Re\scalar{[1-g(h)^{-1}W]\partial_{x_1} W}{\e^{-\i h s}\partial_{x_1}(h)^{-1}W}\int_s^tq_{s_1}(1-\chi_{T_{\rm loc}}(s_1))\d s_1\d s  \\
\Gamma_2&:=Z\int_0^tK(t-s)\Re\scalar{[1-g(h)^{-1}W]\partial_{x_1}W}{\e^{-\i h s}\partial_{x_1}(h)^{-1}W}\d s\int_0^tq_{s_1}(1-\chi_{T_{\rm loc}}(s_1))\d s_1  \\
\Gamma_3&:=ZK(t)\Re\scalar{[1-g(h)^{-1}W]\partial_{x_1}W}{(-\i h)^{-1}\int_0^t[\e^{-\i h(t-s)}-\e^{-\i h t}]q_s(1-\chi_{T_{\rm loc}}(s))\d s \partial_{x_1}(h)^{-1}W}  \,.
\end{align*}

Before the actual estimate we define two continuous functions $\Omega_{1},\ \Omega_{2}:\ (-\infty,1)\rightarrow \mathbb{R}^{+}$ by
\begin{align}
  \Omega_1(\delta)&:=\frac{1}{(1-2\delta)\pi}\int_0^1\frac{1}{1+(1-r)^\pez}(1-r)^\mez[r^\mez-r^{-\delta}]\d r\label{omega1}\\
\Omega_2(\delta)&:=\frac{1}{\pi}\int_0^1\frac{1}{1+(1-r)^\pez}(1-r)^\mez r^{\pez-\delta}\d r\,.\label{omega2}
\end{align}
Recall the function $\Omega(\delta)$ introduced before Theorem \ref{thm:main}. It is given by the sum $\Omega(\delta)=\Omega_1(\delta)+\Omega_2(\delta)$, and we compute
\begin{align*}
\Omega(\delta)=\frac{1}{\pi d(2d-1)}+\frac{1}{2\sqrt{\pi}}\(\frac{2\Gamma(\pez-\delta)}{\Gamma(1-\delta)}-\frac{\Gamma(-\delta)}{\Gamma(\pdz-\delta)}\)\,.
\end{align*}
Note that $\Omega(\delta)$ has only apparent singularities at $\delta=0,\pez$. It is a continuous, monotonically increasing function
\begin{align*}
\Omega:(-\infty,1)\to\RR^+
\end{align*}
with the following properties
\begin{align*}
\lim_{\delta\to -\infty} \Omega(\delta)&=0\\
\Omega(0)&=1-\frac{\log 4}{\pi}\simeq 0.56\\
\Omega(\pez)&=1+\frac{\log 4-2}{\pi}\simeq 0.8\,.\\
\lim_{\delta\to 1}\Omega(\delta)&=\infty
\end{align*}
Numerical analysis suggests that $\Omega(\delta)<1$ for all $\delta<0.66$.

Point (1) of Proposition \ref{prop:contraction-lemma} is covered by the following lemma,
\begin{lem}If $q_t\in B_{\delta,T_{\rm loc}}$ then there is a small constant $\eps(T_{\rm loc})$ satisfying $\eps(\infty)=0$ such that
  \begin{align*}
    |\Gamma_1|&\leq t^{\mez-\delta}[\Omega_1(\delta)+\eps(T_{\rm loc})](1+O(g))\norm{q}_{\delta,T_{\rm loc}}\\
    |\Gamma_3|&\leq t^{\mez-\delta}[\Omega_2(\delta)+\eps(T_{\rm loc})](1+O(g))\norm{q}_{\delta,T_{\rm loc}}\\
    |\Gamma_2|&\leq t^{\mez-\delta}\eps(T_{\rm loc})\norm{q}_{\delta,T_{\rm loc}}\,,
  \end{align*}
\end{lem}
\begin{proof}
We start with $\Gamma_2$. The second term in the product is easy to estimate,
\begin{align}\label{directEst}
|\int_0^tq_{s_1}(1-\chi_{T_{\rm loc}}(s_1))\d s_1|\leq \int_{0}^{t}(1-\chi_{T_{\rm loc}}(s_1)) s_1^{-\frac{1}{2}-\delta}\ ds_1\norm{q}_{\delta,T_{\rm loc}} \lesssim \norm{q}_{\delta,T_{\rm loc}}\,,
\end{align} where we used the fact that $-\frac{1}{2}-\delta>1$, hence $s^{-\frac{1}{2}-\delta}$ is integrable in $[T_{\rm loc},\infty )$.
The first term is estimated as follows. Apply the Fourier transform to the convolution function, then inverse Fourier transform to find
\begin{align*}
Z\int_0^tK(t-s)\Re\scalar{[1-g(h)^{-1}W]\partial_{x_1}W}{\e^{-\i h s}\partial_{x_1}(h)^{-1}W}\d s=\frac{1}{2\pi}\int_{-\infty}^\infty\frac{F(k)}{\i k+ZG(k+\i0)}\e^{-\i kt}\d k \,,
\end{align*}
where $F(k)$ is defined as
\begin{align*}
F(k):&=Z\int_0^\infty\e^{\i ks}\Re \scalar{[1-g(h)^{-1}W]\partial_{x_1}W}{\e^{-\i h s}\partial_{x_1}(h)^{-1}W}\d s\\
&=\frac{Z}{2}\int_0^\infty\e^{\i ks}[\scalar{[1-g(h)^{-1}W]\partial_{x_1}W}{\e^{-\i h s}\partial_{x_1}(h)^{-1}W}+\scalar{[1-g(h)^{-1}W]\partial_{x_1}W}{\e^{\i h s}\partial_{x_1}(h)^{-1}W}]\d s\\
&=\frac{Z}{2}[\scalar{[1-g(h)^{-1}W]\partial_{x_1}W}{(-ih+ik-0)^{-1}\partial_{x_1}(h)^{-1}W}+\scalar{[1-g(h)^{-1}W]\partial_{x_1}W}{(ih+ik-0)^{-1}\partial_{x_1}(h)^{-1}W}]\\
&=-ZG(k+\i0)\,.
\end{align*}
Around $k=0$, the term $\frac{F(k)}{\i k+ZG(k+\i0)}$ has the expansion
\begin{align*}
\frac{F(k)}{\i k+ZG(k+\i0)}=-1+Ck^\pez+O(k)\,.
\end{align*}
The constant term does not contribute, as is seen by integration by parts,
\begin{align*}
  \int_{-\infty}^\infty\frac{F(k)}{\i k+ZG(k+\i0)}\e^{-\i kt}\d k =\int_{-\infty}^\infty\frac{1}{\i t}\partial_k\(\frac{F(k)}{\i k+ZG(k+\i0)}\)\e^{-\i kt}\d k \,,
\end{align*}
and the Fourier transform of $k^\pez$ is of order $t^\mdz$. The detailed computations are identical to \cite{Froehlich102} and thus omitted. We obtain
\begin{align}\label{eq:threehalf}
|Z\int_0^tK(t-s)\Re\scalar{[1-g(h)^{-1}W]\partial_{x_1}W}{\e^{-\i h s}\partial_{x_1}(h)^{-1}W}\d s|\lesssim (1+t)^\mdz\,.
\end{align}
Combining ~\eqref{directEst} and ~\eqref{eq:threehalf}, we have the desired estimate
\begin{align*}
|\Gamma_2|\lesssim (1+t)^\mdz \norm{q}_{\delta,T_{\rm loc}}\leq T_{\rm loc}^\mez(1+t)^{\mez-\delta}\norm{q}_{\delta,T_{\rm loc}}\,.
\end{align*} Here the fact $\delta<1$ has been used.

Now, we turn to $\Gamma_1$. Recall the asymptotic expression for $K$ in Lemma (\ref{eq:k-estimate}),
\begin{align}\label{eq:zKt}
ZK(t)=\frac{3}{\sqrt{2}}\pi^{-\frac{5}{2}}(1+Cg)t^\mez+O(t^{-1})\,,
\end{align}
and that for $\Re\scalar{[1-g(h)^{-1}W]\partial_{x_1} W}{\e^{-\i h t}\partial_{x_1}(h)^{-1}W}$ in ~\eqref{eq:help1} below.

To estimate $\Gamma_1$ we take the leading order terms $\widetilde{K},\widetilde{M},\widetilde{\Gamma}_1$ to approximate the functions in ~\eqref{eq:zKt} and ~\eqref{eq:help1}, and $\Gamma_1$,
\begin{align*}
Z\widetilde{K}(t)&:=\frac{3}{\sqrt{2}}\pi^{-\frac{5}{2}}t^\mez(1+Cg)\\
\widetilde{M}&:=-\frac{1}{3\sqrt{2}}\pi^{\pdz}t^\mdz(1+\widetilde{C}g)\\
\widetilde{\Gamma}_1&:=-Z\int_0^t[\widetilde{K}(t-s)-\widetilde{K}(s)]\widetilde{M}(s)\int_s^tq_{s_1}[1-\chi_T(s_1)]\d s_1\d s\,.
\end{align*}
Compute directly to obtain
\begin{align*}
|\widetilde{\Gamma}_1|&\leq \frac{1+O(g)}{2\pi}\int_0^t[(t-s)^\mez-t^\mez]s^\mdz\int_s^t|q_{s_1}|\d s_1\d s\\
&\leq \frac{1+O(g)}{(1-2\delta)\pi}\int_0^t[(t-s)^\mez-t^\mez]s^\mdz(t^{\pez-\delta}-s^{\pez-\delta})\d s \norm{q_t}_{\delta,T}\\
&=\frac{1+O(g)}{(1-2\delta)\pi}\int_0^t(t-s)^\mez t^\mez\frac{1}{(t-s)^\pez+t^\pez}s^\mez(t^{\pez-\delta}-s^{\pez-\delta})\d s \norm{q_t}_{\delta,T}\,,
\end{align*}
Change variables $s=rt$ to obtain
\begin{align}\label{eq:hatG1}
|\widetilde{\Gamma}_1|\leq t^{\mez-\delta}(1+O(g))\Omega_1(\delta)\norm{q_t}_{\delta,T}\,,
\end{align}
where the constant $\Omega_1$ is defined in (\ref{omega1}).

In what follows we estimate $\Gamma_1-\widetilde{\Gamma}_1$. Divide the integration region $[0,t]$ into three parts, $[0,T_{\rm loc}^{\frac{1}{3}}]$,$[T_{\rm loc}^{\frac{1}{3}},t-T_{\rm loc}^{\frac{1}{3}}]$, and $[t-T_{\rm loc}^{\frac{1}{3}},t]$ and denote these parts by $I_{k},\ k=1,2,3$, i.e.
\begin{align}\label{eq:I123}
\Gamma_1-\widetilde{\Gamma}_1=I_1+I_2+I_3.
\end{align}
We start with estimating $I_1:$
\begin{align*}
I_1:=&Z\int_0^{T_{\rm loc}^{\frac{1}{3}}}[K(t-s)-K(s)]\Re\scalar{[1-g(h)^{-1}W]\partial_{x_1} W}{\e^{-\i h s}\partial_{x_1}(h)^{-1}W}\int_s^tq_{s_1}[1-\chi_{T_{\rm loc}}(s_1)]\d s_1\d s\\
-&Z\int_0^{T_{\rm loc}^{\frac{1}{3}}}[\widetilde{K}(t-s)-\widetilde{K}(s)]\widetilde{M}(s)\int_s^tq_{s_1}[1-\chi_{T_{\rm loc}}(s_1)]\d s_1\d s\,.
\end{align*}
For the term inside the integral we have
\begin{align}\label{eq:KKK}
&|K(t-s)-K(t)|,|\widetilde{K}(t-s)-\widetilde{K}(t)|\nonumber\\
\lesssim &t^\mez(t-s)^\mez\frac{s}{t^\pez+(t-s^\pez)}+(t-s)^\mdz-t^\mdz\\
\lesssim &t^\mdz(1+s)\nonumber
\end{align}
because $s\leq T_{\rm loc}^{\frac{1}{3}}$, and $t\geq T_{\rm loc}$. And consequently
\begin{align*}
|K(t-s)-K(t)|\Re\scalar{[1-g(h)^{-1}W]\partial_{x_1} W}{\e^{-\i h t}\partial_{x_1}(h)^{-1}W}+|\widetilde{K}(t-s)-\widetilde{K}(t)||\widetilde{M}(s)|\lesssim t^\mdz s^\mez(1+O(g))\,.
\end{align*}
Plug this into $I_1$ to obtain
\begin{align}\label{eq:estI1}
|I_1|\lesssim t^{-\frac{3}{2}}(1+O(g))\int_0^{T_{\rm loc}^{\frac{1}{3}}}s^\mez [s^{\frac{1}{2}-\delta}- t^{\frac{1}{2}-\delta}]\d s\norm{q_t}_{\delta,T_{\rm loc}}=t^{-\frac{3}{2}}2T_{\rm loc}^{\frac{1}{6}}\norm{q_t}_{\delta,T_{\rm loc}}\leq t^{\mez-\delta}T_{\rm loc}^{-\frac{1}{3}}\norm{q_t}_{\delta,T_{\rm loc}}\,,
\end{align} where, recall that we only consider the regime $t\geq T_{\rm loc}\gg 1$

Now we turn to $I_2,$ which is defined by
\begin{align*}
I_2:=&Z\int_{T_{\rm loc}^{\frac{1}{3}}}^{t-T_{\rm loc}^{\frac{1}{3}}}[K(t-s)-K(t)]\Re\scalar{[1-g(h)^{-1}W]\partial_{x_1} W}{\e^{-\i h s}\partial_{x_1}(h)^{-1}W} \int_s^t(1-\chi_{T_{\rm loc}}(s_1))q_{s_1}\d s_1\d s\\
-&Z\int_{T_{\rm loc}^{\frac{1}{3}}}^{t-T_{\rm loc}^{\frac{1}{3}}}[\widetilde{K}(t-s)-\widetilde{K}(t)]\widetilde{M}(s)\int_s^t(1-\chi_{T_{\rm loc}}(s_1))q_{s_1}\d s_1\d s\\
=&Z\int_{T_{\rm loc}^{\frac{1}{3}}}^{t-T_{\rm loc}^{\frac{1}{3}}}[(K(t-s)-K(t))-(\widetilde{K}(t-s)-\widetilde{K}(t))]\Re\scalar{[1-g(h)^{-1}W]\partial_{x_1} W}{\e^{-\i h s}\partial_{x_1}(h)^{-1}W}\times\\
 & \int_s^t(1-\chi_{T_{\rm loc}}(s_1))q_{s_1}\d s_1\d s\\
+&Z\int_{T_{\rm loc}^{\frac{1}{3}}}^{t-T_{\rm loc}^{\frac{1}{3}}}[\widetilde{K}(t-s)-\widetilde{K}(t)][\Re\scalar{[1-g(h)^{-1}W]\partial_{x_1} W}{\e^{-\i h s}\partial_{x_1}(h)^{-1}W}-\widetilde{M}(s)]
\int_s^t(1-\chi_{T_{\rm loc}}(s_1))q_{s_1}\d s_1\d s\,.
\end{align*}
For the terms inside the integral we use the following estimates
\begin{align*}
 |K(t-s)-\widetilde{K}(t-s)|\lesssim& (1+t-s)^{-1},\\
|K(t)-\widetilde{K}(t)|\lesssim& (1+t)^{-1},\\
|\widetilde{K}(t-s)-\widetilde{K}(t)|\lesssim &t^{-1}(t-s)^\mez s,\\
\end{align*} implied by ~\eqref{eq:k-estimate} and
\begin{align*}
|\Re\scalar{[1-g(h)^{-1}W]\partial_{x_1} W}{\e^{-\i h s}\partial_{x_1}(h)^{-1}W}-\widetilde{M}(s)|\lesssim &s^{-\frac{5}{2}}\,,
\end{align*} by ~\eqref{eq:help1}
and
\begin{align*}
|\int_{s}^{t} (1-\chi_{T_{\rm loc}}(s_1))q_{s_1}\d s_1|\lesssim \norm{q_t}_{\delta,T_{\rm loc}}[s^{\frac{1}{2}-\delta}-t^{\frac{1}{2}-\delta}].
\end{align*}
to obtain
\begin{align}\label{eq:estI2}
|I_2|\lesssim T_{\rm loc}^{-\frac{1}{3}} t^{\mez-\delta}\norm{q_t}_{\delta,T_{\rm loc}}\,.
\end{align}
For $I_3$, we have
\begin{align*}
\Re\scalar{[1-g(h)^{-1}W]\partial_{x_1} W}{\e^{-\i h s}\partial_{x_1}(h)^{-1}W}\lesssim t^\mdz\,,
\end{align*}
and hence
\begin{align}
|I_3|=&|Z\int_{t-T_{\rm loc}^{\frac{1}{3}}}^t[K(t-s)-K(t)]\Re\scalar{[1-g(h)^{-1}W]\partial_{x_1} W}{\e^{-\i h s}\partial_{x_1}(h)^{-1}W} \int_s^t(1-\chi_{T_{\rm loc}}(s_1))q_{s_1}\d s_1\d s\nonumber\\
-&Z\int_{t-T_{\rm loc}^{\frac{1}{3}}}^t[\widetilde{K}(t-s)-\widetilde{K}(t)]\widetilde{M}(s)\int_s^t(1-\chi_{T_{\rm loc}}(s_1))q_{s_1}\d s_1\d s|\nonumber\\
\lesssim &\int_{t-T_{\rm loc}^{\frac{1}{3}}}^t(|t-s|^\mez+t^\mez)\d s \;t^{-1-\delta}\norm{q_t}_{\delta,T_{\rm loc}}\nonumber\\
\lesssim & T_{\rm loc}^{1/6}t^{-1-\delta}\norm{q_t}_{\delta,T_{\rm loc}}\nonumber\\
\leq &T_{\rm loc}^{-1/3}t^{\mez-\delta}\norm{q_t}_{\delta,T_{\rm loc}}\,\label{eq:estI3}.
\end{align}
Putting ~\eqref{eq:hatG1}, ~\eqref{eq:I123}, ~\eqref{eq:estI1}, ~\eqref{eq:estI2} and ~\eqref{eq:estI3} together, we have shown that
\begin{align*}
|\Gamma_1|\lesssim t^{\mez-\delta}[\Omega_1(\delta)+\eps(T_{\rm loc})](1+O(g))\norm{q_t}_{\delta,T_{\rm loc}}\,,
\end{align*}
where $\eps(T_{\rm loc})\to 0$, as $T_{\rm loc}\to\infty$.\\

Finally, we turn to $\Gamma_3$. Similar to the strategy in estimating $\Gamma_1$, we start with retrieving the main part. Define a new function $\widetilde{V}$ to approximate the function $V(t):=\Re\scalar{[1-g(h)^{-1}W]\partial_{x_1}W}{(-\i h)^{-1}\e^{-\i h t}\partial_{x_1}(h)^{-1}W}$ when $t$ is large,
\begin{align*}
\widetilde{V}:=\frac{\sqrt{2}}{3}\pi^\pdz(1+\widetilde{C}g) t^\mez   \,.
\end{align*} To see this we simply integrate ~\eqref{eq:help1}.

Now, define an approximation $\widetilde{\Gamma}_3$ of $\Gamma_3$,
\begin{align*}
\widetilde{\Gamma}_3:=Z\widetilde{K}(t)\int_0^t[\widetilde{V}(t-s)-\widetilde{V}(t)](1-\chi_{T_{\rm loc}}(s))q_s\d s
\end{align*}
Compute
\begin{align*}
|\widetilde{\Gamma}_3|\leq &t^\mez\frac{1}{\pi}(1+O(g))\int_0^t[(t-s)^\mez-t^\mez]s^{\mez-\delta}\d s\norm{q_t}_{\delta,T_{\rm loc}}\\
=&t^{-1}\frac{1}{\pi}(1+O(g))\int_0^t\frac{s}{(t-s)^\pez+t^\pez}(t-s)^\mez s^{\mez-\delta}\d s \norm{q_t}_{\delta,T_{\rm loc}}\,.
\end{align*}
The change variables $s=tr$ yields
\begin{align}\label{eq:estG3}
|\widetilde{\Gamma}_3|\leq t^{\mez-\delta}\frac{1}{\pi}(1+O(g))\int_0^1(1-r)^\mez\frac{1}{1+(1-r)^\pez}r^{\pez-\delta}\d r \norm{q_t}_{\delta,T_{\rm loc}}=t^{\mez-\delta}\Omega_2(\delta)(1+O(g))\norm{q_t}_{\delta,T_{\rm loc}}\,.
\end{align}
For the difference $\Gamma_3-\widetilde{\Gamma}_3$ we use almost the same techniques in ~\eqref{eq:I123} to obtain
\begin{align}\label{eq:diffG3}
|\Gamma_3-\widetilde{\Gamma}_3|
\lesssim&t^{-1-\delta}(1+O(g)))\norm{q_t}_{\delta,T_{\rm loc}}\leq (1+O(g))T_{\rm loc}^\mez t^{\mez-\delta}\norm{q_t}_{\delta,T_{\rm loc}}\,.
\end{align}
Putting ~\eqref{eq:estG3} and ~\eqref{eq:diffG3} together, we have shown
\begin{align*}
|\Gamma_3|\lesssim t^{\mez-\delta}[\Omega_2(\delta)+\eps(T_{\rm loc})](1+O(g))\norm{q_t}_{\delta,T_{\rm loc}}\,,
\end{align*}
which finishes the proof.
\end{proof}
\section{Proof of Point (2) of Proposition \ref{prop:contraction-lemma}}\label{sec:Point2}
We start with a different results, then show it implies Point (2).
\begin{lem}Let $Q_1$ and $Q_2$ be as in Proposition \ref{prop:contraction-lemma}, and recall the definition of $R(P,t)$ in (\ref{eq:remainder}) and the definitions of $\tilde{D}_{k},\ k=1,2,\cdots,7$ in ~\eqref{eq:remainder}. If $|P_0|\leq T_{\rm loc}^{-2}$ then, for $k=3,4,5,6,7$,
  \begin{align}\label{eq:nonlinear}
   |\widetilde{D}_k(Q_1)-\widetilde{D}_k(Q_2)|(t)\leq t^{-1-\delta}[\eps(T_{\rm loc})+g]\norm{Q_1-Q_2}_{\delta,T_{\rm loc}}(\norm{Q_1}_{\delta,T_{\rm loc}}+\norm{Q_2}_{\delta,T_{\rm loc}})\,.
  \end{align}
\end{lem}
The lengthy proof is divided into the next few subsections.

In the following we use ~\eqref{eq:nonlinear} to prove Point (2):
We start with analyzing the function $K$ of ~\eqref{eq:k}. Lemma ~\ref{lem:k} implies that
\begin{align*}
|K(t-s)-K(t)|\lesssim (1+t-s)^\mez(1+t)^{-1}s+(1+t-s)^{-\frac{3}{2}}+(1+t)^{-\frac{3}{2}}\,.
\end{align*}
Compute directly and use the fact  $|\widetilde{D}_k(Q_1)-\widetilde{D}_k(Q_2)|(t)=0$ for $t\leq T_{\rm loc}$ to obtain the desired estimate
\begin{align}
\int_{0}^{t} |K(t-s)-K(t)| |\widetilde{D}_k(Q_1)-\widetilde{D}_k(Q_2)|(s)\ ds\leq t^{-\frac{1}{2}-\delta}\eps(T_{\rm loc})\norm{Q_1-Q_2}_{\delta,T_{\rm loc}}(\norm{Q_1}_{\delta,T_{\rm loc}}+\norm{Q_2}_{\delta,T_{\rm loc}})\,.
\end{align}

\subsection{The term $\widetilde{D}_3$}
Because of the spherical symmetry of $W$ certain terms in the definition vanish. For notational convenience, we define
\begin{align}
  (1+\frac{g}{\sqrt{\rho_0}}\bar{\beta}^{X_t})\partial_x W^{X_t}:=V^{X_t}\,.
\end{align}
Note that $V$ is rapidly decaying. This makes $\widetilde{D}_3$ take the form
\begin{align*}
\widetilde{D}_3(t)&=-2\sqrt{\rho_0}\Re\scalar{V^{X_t}}{\e^{-\i h^{X_t} t}\beta_0}-2\rho_0\Re\scalar{V^{X_t}}{\e^{-\i h^{X_t} t}\sum_{|\alpha|=2}^{N_0}\frac{1}{\alpha!}(X_0-X_t)^\alpha\partial_x^\alpha(h^{X_t})^{-1}W^{X_t}}\\
&=-2\sqrt{\rho_0}\Re\scalar{\e^{\i h t}V}{\beta_0^{-X_t}}-2\rho_0\Re\scalar{\e^{-\i h^{X_t} t}V^{X_t}}{\sum_{|\alpha|=2}^{N_0}\frac{1}{\alpha!}(X_0-X_t)^\alpha\partial_x^\alpha(h^{X_t})^{-1}W^{X_t}}\\
&=:\widetilde{D}_{31}(X_t)+\widetilde{D}_{32}(X_t)\,.
\end{align*}

Compute directly to obtain
\begin{align*}
|\widetilde{D}_{3,1}(X_t)-\widetilde{D}_{3,1}(\tilde{X}_t)|\leq &|\scalar{\e^{\i h t}V}{(\beta_0^{-X_t}-\beta_0^{-\widetilde{X}_t}}|\\
\leq & t^{-\frac{5}{2}}\norm{\x^{5}(\beta_0^{-X_t}-\beta_0^{-\widetilde{X}_t})}_2\\
\lesssim & t^{-\frac{5}{2}}\norm{\x^{5}(\beta_0-\beta_0^{\int_{T_{\rm loc}}^t (Q_1-Q_2)(s)\d s})}_2\\
=&t^{-\frac{5}{2}}\norm{\x^{5}\int_{T_{\rm loc}}^t\partial_s\beta_0^{\int_{T_{\rm loc}}^s (Q_1-Q_2)(s_1)\d s_1}\d s}_2\\
=&t^{-\frac{5}{2}}\norm{\x^{5}\partial_x\beta_0\cdot\int_{T_{\rm loc}}^t(Q_1-Q_2)(s)\d s}_2\\
\leq& t^{-\frac{5}{2}}\norm{\x^{5}\partial_x\beta_0}_2\int_{T_{\rm loc}}^t|Q_1-Q_2|(s)\d s\\
\lesssim &t^{-1-\delta} \epsilon(T_{loc})\norm{Q_1-Q_2}_{\delta,T_{\rm loc}}
\end{align*}
where, in the second step we used the propagator estimate ~\eqref{eq:nonclass11}, and we used the facts $X_t=X_{T_{\rm loc}}+\int_{T_{\rm loc}}^tQ_1(s)\d s$ and $\widetilde{X}_t=X_{T_{\rm loc}}+\int_{T_{\rm loc}}^tQ_2(s)\d s$, the fact $t\geq T\gg 1$, and $\norm{\x^{3}\partial_x\beta_0}_2\leq 1$ in the Main Theorem ~\ref{thm:main}.

The term $\widetilde{D}_{32}$ is treated similarly:
\begin{align*}
|\widetilde{D}_{32}(Q_1)-\widetilde{D}_{32}(Q_2)|(t)&\lesssim |\scalar{V}{\e^{-\i h t}\sum_{|\alpha|=2}^{N_0}(\widetilde{X}_t-X_t)^\alpha\partial_x^\alpha(h)^{-1}W}|\\
&=|\scalar{V}{\e^{-\i h t}\sum_{|\alpha|=3,\ |\alpha|\ \text{is odd}}^{N_0}(\widetilde{X}_t-X_t)^\alpha\partial_x^\alpha(h)^{-1}W}|\\
&\lesssim t^{-\frac{5}{2}} \sum_{|\alpha|=2}^{N_0}  \int_{t_{\rm loc}}^t|[Q_1(s)-Q_2(s)]^\alpha|\\
&\lesssim t^{-1-\delta}\epsilon(T_{loc})(\norm{Q_1^2}_{\delta,T_{\rm loc}}+\norm{Q_2^2}_{\delta,T_{\rm loc}})\norm{Q_1-Q_2}_{\delta,T_{\rm loc}},
\end{align*} where in the second step the fact that $W$ is spherically symmetric was used to remove the terms with $|\alpha|$ even. In the third step ~\eqref{eq:nonclass2} was used.

So the claim follows for $\widetilde{D}_3$ in ~\eqref{eq:nonlinear}.

\subsection{The term $\widetilde{D}_4$ and $\widetilde{D}_5$}
In what follows we only estimate $\widetilde{D}_4$. The analysis of $\widetilde{D}_5$ is almost identical, hence omitted.

Consider
\begin{align*}
\widetilde{D}_4(t)=&-\i g\scalar{V^{X_t}}{\int_0^t\e^{-\i h^{X_t}(t-s)}[W^{X_s}-W^{X_t}]\delta_s\d s}\\
=&-\i g\scalar{V}{\int_0^t\e^{-\i h (t-s)}[W^{X_s-X_t}-W]\delta_s^{-X_t}\d s}
\end{align*}
Observe that the term depends on $Q$, namely $(X_s-X_t)$ and $\delta_s^{-X_t}$. Consequently
\begin{align*}
|\widetilde{D}_{4}(Q_1)-\widetilde{D}_{4}(Q_2)|(t)
&\lesssim g|\scalar{V}{\int_0^t\e^{-\i h(t-s)} [W^{X_s-X_t}-W^{\widetilde{X}_s-\widetilde{X}_t}]\delta^{-X_t}_s\d s}|\\
&+g|\scalar{V}{\int_0^t\e^{-\i h(t-s)} [W^{\widetilde{X}_s-\widetilde{X}_t}-W](\delta^{-X_t}_s-\widetilde{\delta}^{-\widetilde{X}_t}_s)\d s}|\\
&=:B_1+B_2
\end{align*}
where $B_1$ and $B_2$ are defined naturally.
For $B_1$, by the decay estimate in Proposition ~\ref{eq:classical}, we obtain
\begin{align*}
B_1&\lesssim g\int_0^t|\widetilde{X}_t-\widetilde{X}_s+X_s-X_t|(1+t-s)^{-\frac{3}{2}}\norm{\x^3\partial_xW\delta^{-X_t}_s}_2\d s\\
&\lesssim  g\int_0^t|\widetilde{X}_t-\widetilde{X}_s+X_s-X_t|(1+t-s)^{-\frac{3}{2}}(1+s)^\mez\d s\,.
\end{align*}
Both $Q_1$ and $Q_2$ are in $B_{\delta,T_{\rm loc}}$, and thus
\begin{align*}
|\widetilde{X}_t-\widetilde{X}_s+X_s-X_t|&=|\int_s^t[Q_1-Q_2](s_1)\d s_1|\leq \norm{Q_1-Q_2}_{\delta,T_{\rm loc}}\int_s^ts_1^{\mez-\delta}\d s_1\\
&\lesssim\norm{Q_1-Q_2}_{\delta,T_{\rm loc}}(t^{\pez-\delta}-s^{\pez-\delta}).
\end{align*}

Apply Lemma ~\ref{LM:convo} to obtain
\begin{align}\label{eq:mira}
B_{1}(t) \lesssim g t^{-1-\delta}\norm{Q_1-Q_2}_{\delta,T_{\rm loc}}\,.
\end{align}
For $B_2$, since $\delta_s$ depends on $Q$, we have
\begin{align}
B_2
\lesssim& g \int_0^t(1+t-s)^{-\frac{3}{2}}|\xt_t-\xt_s|\norm{\x^3\partial_xW(\delta^{-X_t}_s-\widetilde{\delta}^{-\widetilde{X}_t}_s)}_2\d s\nonumber\\
\lesssim & g \|Q_2\|\int_0^t(1+t-s)^{-\frac{3}{2}}((1+s)^{\frac{1}{2}-\delta}-(1+t)^{\frac{1}{2}-\delta})\norm{\x^3\partial_xW(\delta^{-X_t}_s-\widetilde{\delta}^{-\widetilde{X}_t}_s)}_2\d s\nonumber
\end{align}
In Lemma ~\ref{LM:diff} below we prove
\begin{align}\label{eq:claim1}
\norm{\x^3\partial_xW(\delta^{-X_t}_s-\widetilde{\delta}^{-\widetilde{X}_t}_s)}_2\lesssim (1+s)^{-\frac{1}{2}-\delta}\norm{Q_1-Q_2}_{\delta,T_{\rm loc}}\,
\end{align} which together with a weaker version of Lemma ~\ref{LM:convo} implies the desired estimate.

The proof is complete.
\begin{flushright}
$\square$
\end{flushright}

The following result has been used in ~\eqref{eq:claim1}.
\begin{lem}\label{LM:diff}
\begin{align}
\norm{\x^6\partial_xW(\delta^{-X_t}_s-\widetilde{\delta}^{-\widetilde{X}_t}_s)}_2\lesssim (1+s)^{-\frac{1}{2}-\delta}\norm{Q_1-Q_2}_{\delta,T_{\rm loc}}.
\end{align}
\end{lem}
\begin{proof}
Recall the definition of $\delta_{t,T}, \ t\leq T,$ from ~\eqref{eq:decomp}.
We start with deriving an equation for $\delta^{-X_{t}}_s=\delta^{-X_{t}}_{s,t}, \ s\leq  t,$ from ~\eqref{eq:deltaduha2}:
\begin{align}
&\delta_{s,t}^{-X_t}&&=&&\e^{-\i h s}\sqrt{\rho_0}(X_0-X_t)\cdot\partial_x h^{-1} W-\frac{\sqrt{\rho_0}}{M}\int_0^s\e^{-\i h (s-s_1)}P_{s_1}\cdot\partial_x h^{-1} W \d s_1\nonumber\\
&&&&&+\e^{-\i h s}[\beta_0^{-X_t}-\bar{\beta}+\sqrt{\rho_0}\sum_{|\alpha|=2}^{N_0}\frac{1}{\alpha!}(X_0-X_t)^\alpha\partial_x^\alpha h^{-1} W]\label{eq:diff2}\\
&&&&&-\i g\int_0^s \e^{-\i h(s-s_1)}[W^{X_s-X_t}-W]\delta_{s_1,t}^{-X_{t}}\d s_1\nonumber\\
&&&&&-\frac{\sqrt{\rho_0}}{M}\sum_{|\alpha|=2}^{N_0}\frac{1}{\alpha!}\int_0^s \e^{-\i h (s-s_1)}\partial_x^\alpha h^{-1} W P_{s_1}\alpha(X_{s_1}-X_T)^{\alpha-1}\d s_1
+\i\int_0^s \e^{-\i h (s-s_1)}[G_1(s_1)]^{-X_t}\d s_1\nonumber
\end{align}
Introduce a new function $\eta$,
\begin{align*}
\eta_s:=\delta^{-X_t}_s-\widetilde{\delta}^{-\widetilde{X}_t}_s\,,
\end{align*} which satisfies a new equation
\begin{align}
\eta_s=&\e^{-\i h s}\sqrt{\rho_0}[(X_0-X_t)-(\tilde{X}_0-\tilde{X}_t)]\cdot\partial_x h^{-1} W-\frac{\sqrt{\rho_0}}{M}\int_0^s\e^{-\i h (s-s_1)}[Q_1(s_1)-Q_2(s_1)]\cdot\partial_x h^{-1} W \d s_1\\ \label{eq:etas}
&+\e^{-\i h s}[\beta_0^{-X_t}-\beta_0^{-\tilde{X}_t}]+\sqrt{\rho_0}\sum_{|\alpha|=2}^{N_0}\frac{1}{\alpha!}[(X_0-X_t)^\alpha-(\tilde{X}_0-\tilde{X}_t)^{\alpha}]\partial_x^\alpha h^{-1} W]\nonumber\\
&-\i g\int_0^s \e^{-\i h(s-s_1)}[W^{X_s-X_t}-W]\eta_{s_1}\d s_1-\i g\int_0^s \e^{-\i h(s-s_1)}[W^{X_s-X_t}-W^{\tilde{X_s}-\tilde{X_t}}]\widetilde{\delta}_{s_1,t}^{-X_{t}}\d s_1\nonumber\\
&-\frac{\sqrt{\rho_0}}{M}\sum_{|\alpha|=2}^{N_0}\frac{1}{\alpha!}\int_0^s \e^{-\i h (s-s_1)}\partial_x^\alpha h^{-1} W [Q_1(s_1)-Q_2(s_1)]\alpha(X_{s_1}-X_T)^{\alpha-1}\d s_1\nonumber\\
&-\frac{\sqrt{\rho_0}}{M}\sum_{|\alpha|=2}^{N_0}\frac{1}{\alpha!}\int_0^s \e^{-\i h (s-s_1)}\partial_x^\alpha h^{-1} W Q_{2}(s_1)\alpha [(X_{s_1}-X_T)^{\alpha-1}-(\tilde{X}_{s_1}-\tilde{X}_T)^{\alpha-1}]\d s_1\nonumber\\
&+\i\int_0^s \e^{-\i h (s-s_1)}{[G_1^{-X_t}(s_1)-\tilde{G}_1^{-\tilde{X}_t}(s_1)]}\d s_1.
\end{align}

Here the function $\tilde{G}_1$ is defined in the same way as $G_1$ (see ~\eqref{eq:G1}), the only difference is that it depends on the new trajectory $\tilde{X}.$

What is left is essentially to improve the proof of Proposition ~\ref{prop:delta}. The difference between ~\eqref{eq:etas} and ~\eqref{eq:deltaduha2} is that the term $-e^{ih^{X_{T}} t}\bar{\beta}^{X_{T}}$ is not present. This makes the proof almost identical to that of Corollary ~\ref{cor:oneterm}, hence we omit the details here.

The proof is complete.
\end{proof}

\subsection{The term $\widetilde{D}_6$}
Consider finally
\begin{align*}
\i\scalar{V^{X_t}}{\int_0^t\e^{-\i h^{X_t} (t-s)}G_1(s)\d s}\,,
\end{align*}
so that
\begin{align*}
|\widetilde{D}_{6}(Q_1)-\widetilde{D}_{6}(Q_2)|(t)=&|\scalar{\partial_xW}{\int_0^t\e^{-\i h (t-s)}(h^{X_s-X_t}r^{-X_t}_{N_0}-h^{\xt_s-\xt_t}\widetilde{r}^{-\xt_t}_{N_0})\d s}|\\
\leq &|\scalar{\partial_xW}{\int_0^t\e^{-\i h (t-s)}(h^{X_s-X_t}-h^{\xt_s-\xt_t})r^{-X_t}_{N_0}(s)\d s}|\\
&+|\scalar{\partial_xW}{\int_0^t\e^{-\i h (t-s)}h^{\xt_s-\xt_t}(r^{-X_t}_{N_0}-\widetilde{r}^{-\xt_t}_{N_0})(s)\d s}|\\
=&|\scalar{\e^{-\i h (t-s)}\partial_xW}{\int_0^t(h^{X_s-X_t}-h^{\xt_s-\xt_t})r^{-X_t}_{N_0}(s)\d s}|\\
&+|\scalar{\e^{-\i h (t-s)}\partial_xW}{\int_0^t h^{\xt_s-\xt_t}(r^{-X_t}_{N_0}-\widetilde{r}^{-\xt_t}_{N_0})(s)\d s}|\\
=&D_1+D_2.
\end{align*}

For $D_1$, use
\begin{align*}
h^{X_s-X_t}-h^{\xt_s-\xt_t}=g(W^{X_s-X_t}-W^{\xt_s-\xt_t})
\end{align*} and
\begin{align*}
\|\langle x\rangle^{-5}\e^{-\i h t}\partial_xW\|_2\lesssim (1+|t|)^{-\frac{5}{2}}
\end{align*} of ~\eqref{eq:nonclass11}
to obtain
\begin{align*}
D_1\lesssim &g\int_0^t(1+t-s)^{-\frac{5}{2}}\norm{\x^5(W^{X_s-X_t}-W^{\xt_s-\xt_t}) r^{-X_t}_{N_0}}_2\d s\\
\lesssim &g\int_0^t(1+t-s)^{-\frac{5}{2}}(|X_s-X_t|)\norm{\x^{-5}r^{-X_t}_{N_0}}_2\d s\\
\lesssim &g\norm{Q_1-Q_2}_{T_{\rm loc},\delta}\int_0^t(1+t-s)^{-\frac{5}{2}}(1+s)^{\pez-\delta}(1+s)^\mdz\d s\\
\lesssim &g\norm{Q_1-Q_2}_{T_{\rm loc},\delta}t^{-1-\delta}\,,
\end{align*}
where in the second step we used
\begin{align*}
\norm{\x^{-3}r_{N_0}(s)}_2\lesssim (1+s)^\mdz\,
\end{align*} of ~\eqref{eq:rN2}.

Now we turn to $D_2.$
The key part is to estimate $M(t,s):=h^{\xt_s-\xt_t}(r^{-X_t}_{N_0}-\widetilde{r}^{-\xt_t}_{N_0}).$ By methods similar to the proof of ~\eqref{eq:rN} we obtain
\begin{align}\label{eq:claim}
\| \langle x\rangle^{5}M(t,s)\|_{2}\lesssim (1+s)^{-1-\delta}(\norm{Q_1}_{\delta,T_{\rm loc}}+\norm{Q_2}_{\delta,T_{\rm loc}})\norm{Q_1-Q_2}_{\delta,T_{\rm loc}}.
\end{align} This together with ~\eqref{eq:nonclass11}
implies that
\begin{align}
|D_2|\lesssim \int_{0}^{t}(1+t-s)^{-\frac{5}{2}}(1+s)^{-1-\delta}(\norm{Q_1}_{\delta,T_{\rm loc}}+\norm{Q_2}_{\delta,T_{\rm loc}})\norm{Q_1-Q_2}_{\delta,T_{\rm loc}}.
\end{align}
\subsection{The Term $\tilde{D}_{7}$}
The last step is to incorporate the term $g\scalar{\partial_xW^{X_t}}{|\delta_t|^2}=
g\scalar{\partial_xW}{|\delta_t^{-X_t}|^2}$. Compute directly to obtain
\begin{align*}
|\delta_t^{-X_t}(Q_1)|^2-|\delta_t^{-\tilde{X}_t}(Q_2)|^2=(\delta_t^{-X_t}(Q_1)-\delta_t(Q_2))[\delta^{-X_t}_t]^*(Q_1)
+\delta^{-\tilde{X}_t}_t(Q_2)([\delta^{-X_t}_t]^*(Q_1)-[\delta^{-\tilde{X}_t}_t]^*(Q_2))\,.
\end{align*}
Apply Lemma ~\ref{LM:diff} and use the estimate $\|\langle x\rangle^{-3}\delta_t^{-X_t}\|_2, \ \|\langle x\rangle^{-3}\delta_t^{-\tilde{X}_t}\|_2\lesssim (1+t)^{-\frac{1}{2}}$ of Proposition ~\ref{prop:delta} to obtain the desired estimate
\begin{align*}
g|\scalar{\partial_xW^{X_t}}{|\delta_t(Q_1)|^2-|\delta_t(Q_2)|^2}|&\lesssim g\norm{\x^3\partial_xW^{X_t}\delta_t(Q_1)-\delta_t(Q_2)}_2|\norm{\x^{-3}\delta_t^*(Q_1)}_2\\
&\lesssim g(\norm{Q_1}_{T_{\rm loc},\delta}+\norm{Q_2}_{T_{\rm loc},\delta})\norm{Q_1-Q_2}_{T_{\rm loc},\delta}t^{-\delta-1}\,.
\end{align*}

\section{Proof of Proposition \ref{prop:contraction-lemma2}}
As discussed at the beginning of Section ~\ref{sec:Point2} it is sufficient for us to prove the following result.
\begin{lem}
\begin{align*}
|A(\chi_{T_{\rm loc}}(P))|\leq \eps(T_{\rm loc})t^{\mez-\delta}\,.
\end{align*}
\end{lem}
\begin{proof} Recall the definition of $A$ in (\ref{eq:defina}) and the local existence estimate $|P_t|\leq T_{\rm loc}^{-2}$ for $t\in [0,T_{\rm loc}]$. Then compute
  \begin{align*}
  |A(\chi_{T_{\rm loc}} P)|\leq &Z\eps(T_{\rm loc})[\int_0^{T_{\rm loc}}|K(t-s)-K(t)||\Re\scalar{[1-g(h)^{-1}W]\partial_{x_1} W}{\e^{-\i h s}\partial_{x_1}(h)^{-1}W}|\d s\\
+&|\int_0^tK(t-s)\Re\scalar{[1-g(h)^{-1}W]\partial_{x_1}W}{\e^{-\i h s}\partial_{x_1}(h)^{-1}W}\d s|\\
+&|K(t)|\int_0^{T_{\rm loc}}|\Re\scalar{[1-g(h)^{-1}W]\partial_{x_1}W}{(-\i h)^{-1}[\e^{-\i h(t-s)}-\e^{-\i h t}]\partial_{x_1}(h)^{-1}W}\d s]\,.
  \end{align*}
As proved in (\ref{eq:threehalf}) the second term on the right hand side is of order $t^\mdz$. For the third term, we have by a computation similar to (\ref{eq:KKK})
\begin{align*}
|\Re\scalar{[1-g(h)^{-1}W]\partial_{x_1}W}{(-\i h)^{-1}[\e^{-\i h(t-s)}-\e^{-\i h t}]\partial_{x_1}(h)^{-1}W}|\lesssim (1+t-s)^\mez(1+t)^{-1}s+(1+t-s)^\mdz\,.
\end{align*}
So we obtain
\begin{align*} |A(\chi_{T_{\rm loc}} P)|\lesssim &\eps(T_{\rm loc})[(1+t)^{-1}\int_0^{T_{\rm loc}}(1+t-s)^\mez\frac{s}{(1+s)^\pdz}\d s+\int_0^{T_{\rm loc}}(1+t-s)^\mdz(1+s)^\mdz\d s\\
&+(1+t)^\mdz+(1+t)^\mdz\int_0^{T_{\rm loc}}(1+t-s)^\mez s\d s+(1+t)^\mez\int_0^{T_{\rm loc}}(1+t-s)^\mdz\d s]\\
&\leq \eps(T_{\rm loc})(1+t)^{\mez-\delta}\,,
\end{align*}
where $\eps(T_{\rm loc})\to 0$ as $T_{\rm loc}\to\infty$.
\end{proof}
We are left with proving
\begin{align*}
\int_0^t[K(t-s)-K(t)]\widetilde{D}_k(\chi_{T_{\rm loc}}P,s)\d s\big|\leq t^{\mez-\delta}\eps(T_{\rm loc})\,,
\end{align*}
for $k=3,4,5,6,7$, recall that we only consider $t\geq T_{loc}\gg 1$. As in the proof of Proposition \ref{prop:contraction-lemma}, all we have to show is
\begin{align*}
|\widetilde{D}_k(\chi_{T_{\rm loc}}P,s)|\leq (1+t)^{-1-\delta}\eps(T_{\rm loc})\,.
\end{align*}
The estimates are very similar to the ones in the proof of Proposition \ref{prop:contraction-lemma}, so we will do only three of them, namely $\widetilde{D}_3,$ $\widetilde{D}_4$ and $\widetilde{D}_7.$

Apply propagator estimates in Proposition ~\ref{prop:propagator} to estimate $\widetilde{D}_3$
\begin{align*}
|\widetilde{D}_3(\chi_{T_{\rm loc}}P,t)|&\lesssim |\scalar{\partial_xW^{X_t}}{\e^{-\i h^{X_t} t}\beta_0}|+|\scalar{\partial_xW^{X_t}}{\e^{-\i h^{X_t} t}\sum_{|\alpha|=2}^{N_0}\frac{1}{\alpha!}(X_0-X_t)^\alpha\partial_x^\alpha( h^{X_t})^{-1}W^{X_t}}|\\
\lesssim& (1+\|\langle x\rangle^{5}\beta_0\|_2) t^{-\frac{5}{2}}\leq (1+t)^{-1-\delta}\eps(T_{\rm loc}) (1+\|\langle x\rangle^{5}\beta_0\|_2)\,.
\end{align*}

For $\widetilde{D}_4$, observe
\begin{align*}
|\widetilde{D}_4(\chi_{T_{\rm loc}}P,t)|&\lesssim g |\scalar{\partial_xW^{X_t}}{\int_0^{T_{\rm loc}}\e^{-\i h^{X_t}(t-s)}\partial_xW^{X_t}\cdot (X_s-X_{T_{\rm loc}})\delta_s\d s}|\\
&\lesssim g \int_0^{T_{\rm loc}}(1+t-s)^{-\frac{5}{2}}\int_s^{T_{\rm loc}}|P_{s_1}|\d s_1\norm{\x^{-3}\delta_s^{-X_t}}_2\d s\,.
\end{align*}
Because of the local existence estimate \eqref{eq:localdecay}, we have
\begin{align*}
\int_s^{T_{\rm loc}}|P_{s_1}|\d s_1\leq T_{\rm loc}^{-2}(T_{\rm loc}-s)\,.
\end{align*}
So we can estimate
\begin{align*}
|\widetilde{D}_4(\chi_{T_{\rm loc}}P,t)|&\lesssim g T_{\rm loc}^{-2}\int_0^{T_{\rm loc}}(1+t-s)^{-\frac{5}{2}}(T_{\rm loc}-s)(1+s)^\mez\d s\lesssim gT_{\rm loc}^\mez (1+t)^{-\frac{1}{2}-\delta}\,.
\end{align*}
Here in the last inequality two regimes have been considered: $t\in [T_{\rm loc},2T_{\rm loc}]$ and $t>2T_{\rm loc}.$ In the first one we bound $(1+t-s)^{-\frac{5}{2}}(T_{\rm loc}-s)$ by $(1+t-s)^{-\frac{3}{2}};$ and in the second bound $(1+t-s)^{-\frac{5}{2}}(T_{\rm loc}-s)$ by $t^{-\frac{5}{2}}T_{\rm loc}$ for $s\in [0,T_{\rm loc}].$

Now we turn to $\widetilde{D}_7.$
We estimate the term $g\scalar{\partial_xW^{X_t}}{|\delta_t|^2(\chi_{T_{\rm loc}}P)}$ by applying Corollary ~\ref{cor:oneterm}. Observe that $\partial_xW^{X_t}\perp |e^{-h^{X_{T}}}\bar{\beta}^{X_{T}}|^2$ by the fact the latter is spherically symmetric. This together with the estimate on $\phi_t$ and the fact $\|\langle x\rangle^{-3}e^{-h^{X_{T}}}\bar{\beta}^{X_{T}}\|_{2}\lesssim (1+t)^{-\frac{1}{2}}$, implies the desired estimate
\begin{align*}
|g\scalar{\partial_xW^{X_t}}{\delta_t(\chi_{T_{\rm loc}}P)\delta^*_t(\chi_{T_{\rm loc}}P)}|\lesssim g (1+t)^{-1-\delta}\,.
\end{align*}
\hfill $\square$

\section{Propagator Estimates}\label{sec:propagator-estimates}
In this section we prove the propagator estimates used throughout the article. Define $h:=-\Delta+gW$, with $g\in\RR$ small and $W(x)=W(|x|)$.
\begin{prop}\label{prop:propagator}If $W:\RR^3\to\RR$ is a smooth function and decays exponentially fast at $\infty$ we have
\begin{align}\label{eq:classical}
 \norm{\x^{-3}\e^{\i th}(h)^{-1+\eps}\phi}_2&\leq C(1+t)^{\mez(1+2\eps)}\norm{\x^3\phi}_2\,,\quad \eps\in[0,1]\\
\norm{\x^{-5}\e^{\i th}\partial_x(h)^{-1}W}_2&\leq C(1+t)^\mdz\label{eq:nonclass1}\\
\norm{\x^{-5}\e^{\i th}\partial_xW}_2&\leq C(1+t)^{-\frac{5}{2}}\label{eq:nonclass11}\\
\norm{\x^{-5}\e^{\i th}\partial_x^\alpha(h)^{-1}W}_2&\leq C(1+t)^{-\frac{5}{2}}\label{eq:nonclass2}\,,\quad |\alpha|\geq 3,\ \alpha\ \text{is odd}\,.
\end{align}
\end{prop}
Estimate ~\eqref{eq:classical} is a classic result, see e.g.\ \cite{jensen79}. In the proof of the remaining assertions we will use the following
\begin{lem}\label{lem:prop}For any smooth, spherically symmetric and fast decaying function $\varphi$ we have
  \begin{align*}
    \norm{\x^{-5}\e^{-\i t\Delta}\partial_x(-\Delta)^{-1}\varphi}_2\leq C(1+t)^\mdz\norm{\x^4\varphi}_2\,.
  \end{align*}
\end{lem}
\begin{proof}
By Fourier transform we obtain
\begin{align*}
\e^{-\i t\Delta}\partial_x(-\Delta)^{-1}\varphi=C\int_{\RR^3}\e^{\i k\cdot x}\e^{\i t|k|^2}\frac{k}{|k|^2}\hat{\varphi}(k)\d k\,,
\end{align*}
for some constant $C\in \RR$. Since $\varphi$ is spherically symmetric, so is $\hat{\varphi}$. Using polar coordinates ($\RR^3\ni k=\rho g(\alpha,\vartheta)$) we find
\begin{align*}
\e^{-\i t\Delta}\partial_x(-\Delta)^{-1}\varphi&=C\int_{-1}^1\int_0^{2\pi}\int_0^\infty\e^{\i\rho |x|\cos\vartheta} \e^{\i t\rho^2}\rho g(\alpha,\vartheta)\hat{\varphi}(\rho)\d\rho\d\alpha\d\cos\vartheta\\
&=C\int_{-1}^1\int_0^{2\pi}\int_0^\infty\e^{\i t\rho^2}\rho g(\alpha,\vartheta)\hat{\varphi}(\rho)\d\rho\d\alpha\d\cos\vartheta\\
&+C\int_{-1}^1\int_0^{2\pi}\int_0^\infty\frac{\e^{\i\rho |x|\cos\vartheta}-1}{\rho}\rho^2\e^{\i t\rho^2}g(\alpha,\vartheta)\hat{\varphi}(\rho)\d\rho\d\alpha\d\cos\vartheta\\
&=C\int_{-1}^1\int_0^{2\pi}\int_0^\infty\frac{\e^{\i\rho |x|\cos\vartheta}-1}{\rho}\rho^2\e^{\i t\rho^2}g(\alpha,\vartheta)\hat{\varphi}(\rho)\d\rho\d\alpha\d\cos\vartheta\,,
\end{align*}
as the unit vector $g(\alpha,\vartheta)$ averages to zero over the unit sphere. Denote by $f_x(\rho)$ the smooth function $\frac{\e^{\i\rho |x|\cos\vartheta}-1}{\rho}$ and evaluate the $\rho$-integral by scaling $\rho\to t^\pez\rho$ as follows:
\begin{align}\label{eq:rho1}
\int_0^\infty f_x(\rho)\rho^2\e^{\i t\rho^2}\hat{\varphi}(\rho)\d\rho=t^\mdz\int_0^\infty f_x(\rho t^\mez)\rho^2\e^{\i \rho^2}\hat{\varphi}(\rho t^\mez)\d\rho  \,.
\end{align}
Since $\rho^2\e^{\i \rho^2}$ is not integrable we integrate by parts which yields
\begin{align*}
  &t^\mdz\int_0^\infty f_x(\rho t^\mez)\rho^2\e^{\i \rho^2}\hat{\varphi}(\rho t^\mez)\d\rho=-t^\mdz\frac{1}{2\i}\int_0^\infty \e^{\i \rho^2}\partial_\rho(\rho f_x(\rho t^\mez)\hat{\varphi}(\rho t^\mez))\d\rho\\
=&-t^\mdz\frac{1}{2\i}\int_0^\infty \e^{\i \rho^2} f_x(\rho t^\mez)\hat{\varphi}(\rho t^\mez)\d\rho-t^\mdz \frac{1}{2\i}\int_0^\infty \e^{\i \rho^2}\rho \partial_\rho (f_x(\rho t^\mez)\hat{\varphi}(\rho t^\mez))\d\rho\,.
\end{align*}
The first term on the second line is easily seen to be given by
\begin{align*}
t^\mdz\frac{1}{2\i}\int_0^\infty \e^{\i \rho^2} f_x(\rho t^\mez)\hat{\varphi}(\rho t^\mez)\d\rho=Ct^\mdz (f_x(0)\hat{\varphi}(0)+o(1))=Ct^\mdz (\cos\vartheta|x|\hat{\varphi}(0)+o(1))\,,
\end{align*}
where $o(1)$ is short for $o(1)\,,\; t\to\infty$. In the second term we integrate by parts again to get
\begin{align*}
&t^\mdz \frac{1}{2\i}\int_0^\infty \e^{\i \rho^2}\rho \partial_\rho (f_x(\rho t^\mez)\hat{\varphi}(\rho t^\mez))\d\rho\\
=&-Ct^\mdz (f'_x(0)\hat{\varphi}(0)+f_x(0)\hat{\varphi}'(0))
-t^\mdz\frac{1}{2\i} \int_0^\infty \e^{\i \rho^2}\partial^2_\rho (f_x(\rho t^\mez)\hat{\varphi}(\rho t^\mez))\d\rho  \,.
\end{align*}
The last term is given by
\begin{align*}
t^\mdz\frac{1}{2\i} \int_0^\infty \e^{\i \rho^2}\partial^2_\rho (f_x(\rho t^\mez)\hat{\varphi}(\rho t^\mez))\d\rho=Ct^\mdz (\partial^2_r\bigg|_{r=0} f_x(r)\hat{\varphi}(r)+o(1))  \,.
\end{align*}
Summarizing, we have shown
\begin{align*}
|\e^{-\i t\Delta}\partial_x(-\Delta)^{-1}\varphi (x)|\leq Ct^\mdz (|x|^3 \norm{y^2\varphi}_1+o(1))\,,
\end{align*}
because $f''_x(0)=-\i\cos^3\vartheta|x|^3$, and $\hat{\varphi}''(0)=\int y^2\varphi \d^3 y$. Using H\"older's inequality we arrive at
\begin{align*}
\norm{\x^{-5}\e^{-\i t\Delta}\partial_x(-\Delta)^{-1}\varphi}_2\leq C t^\mdz\norm{\x^4\varphi}_2\,,
\end{align*}
which is the claim.
\end{proof}

\begin{proof}[Proof of Proposition \ref{prop:propagator}]
The proof of (\ref{eq:nonclass11}) uses the spherical symmetry of $W$ and the Kato-Jensen expansion of the propagator \cite{jensen79},
\begin{align*}
\e^{-\i ht}=t^\mdz B_1+t^{-\frac{5}{2}}B_2+\dots
\end{align*}
in $\mathcal{B}(L^{2,-5},L^{2,5})$ (this expansion holds if $h=-\Delta+gW$ has no negative eigenvalues and no zero resonance, which is the case for our choice of $gW$). Here,
\begin{align*}
B_1(\cdot) = C\scalar{\cdot}{(1+(-\Delta)^{-1}gW)^{-1}1}(1+(-\Delta)^{-1}gW)^{-1}1\,,
\end{align*}
so $B_1(\partial_xW)=0$ because of the spherical symmetry of $W$. We obtain
\begin{align*}
\norm{\x^{-5}\e^{\i ht}\partial_xW}_2\leq C(1+t)^{-\frac{5}{2}}\,,
\end{align*}
which is (\ref{eq:nonclass11}).

To prove (\ref{eq:nonclass1}) define the function
\begin{align*}
\xi:=(1+gW(-\Delta)^{-1})^{-1}W\,.
\end{align*}
The function $\xi$ is spherically symmetric, and from the equation $(-\Delta)^{-1}\xi=h^{-1}W$ we get
\begin{align*}
\xi=(-\Delta)h^{-1}W=W-gWh^{-1}W\,,
\end{align*}
from wich it is easy to see that $\xi$ decays exponentially fast at $\infty$, since $h^{-1}$ is a bounded operator $\mathscr{H}^{2,s}\to \mathscr{H}^{2,-s}$ for $s>\frac{1}{2}$.\\
By Duhamel's principle, we rewrite $\e^{\i th}\partial_xh^{-1}W$ as
\begin{align*}
\e^{\i th}\partial_xh^{-1}W &= \e^{\i th}\partial_x(-\Delta)^{-1}\xi\\
&=\e^{-\i t\Delta}\partial_x(-\Delta)^{-1}\xi+\int_0^t\e^{\i h(t-s)}gW\e^{-\i s\Delta}\partial_x(-\Delta)^{-1}\xi\d s\,.
\end{align*}
The desired estimate follows from (\ref{eq:classical}) and Lemma \ref{lem:prop}.
\end{proof}

The proof for (\ref{eq:nonclass2}) proceeds along the very same lines as the one for (\ref{eq:nonclass1}). The integral over the angle function $g(\phi,\vartheta)$ ($\RR^3\ni k^\alpha=\rho^{|\alpha|}g(\phi,\vartheta)$) still vanishes because $|\alpha|$ is odd, and the additional (at least) two powers of $\rho$ in (\ref{eq:rho1}) give the desired $t^{-\frac{5}{2}}$-decay. The increased number of integrations by parts poses no problem since $\norm{\x^nW}_2<\infty$ for any $n$.

The proof is complete.
\begin{flushright}
$\square$
\end{flushright}

The following result has been used several times.
\begin{lem}
\begin{align}\label{eq:help1}
\Re\scalar{[1-g(h)^{-1}W]\partial_{x_1} W}{\e^{-\i h t}\partial_{x_1}(h)^{-1}W}&=\frac{1}{3}\Re \scalar{W}{\e^{\i\Delta t}W}+O(gt^\mdz)\nonumber\\
&=-\frac{1}{3\sqrt{2}}\pi^{\pdz}t^\mdz(1+\widetilde{C}g)+O(t^{-\frac{5}{2}})\,.
\end{align}
\end{lem}
\begin{proof}
The results in ~\cite{jensen79} implies that the expression is of the form $C_1 t^{-\frac{1}{2}}+C_2 t^{-\frac{3}{2}}+O( t^{-\frac{5}{2}})$ as $t\rightarrow \infty.$ It is very involved to compute each constant. Instead in what follows we show that $C_1=0$ and $C_2=-\frac{1}{3\sqrt{2}}\pi^{\pdz}+O(g).$

To prove this we expand $h^{-1}$ to be $(-\Delta)^{-1}+g(-\Delta)^{-1}W h^{-1}$ and obtain
\begin{align*}
\scalar{[1-g(h)^{-1}W]\partial_{x_1} W}{\e^{-\i h t}\partial_{x_1}(h)^{-1}W}&=\scalar{\partial_{x_1}W}{\e^{\i\Delta t}\partial_{x_1}(-\Delta)^{-1}W}-g\scalar{(h^{-1}W)\partial_{x_1} W}{\e^{\i\Delta t}\partial_{x_1}(-\Delta)^{-1}W}\\
-&g\scalar{V}{\e^{\i\Delta t}\partial_{x_1}(-\Delta)^{-1}Wh^{-1}W}+g\scalar{V}{\int_0^t\e^{\i\Delta(t-s)}W\e^{-\i hs}\partial_{x_1}h^{-1}W}\,,
\end{align*}
where we used the abbreviation
\begin{align*}
V=[1-g(h)^{-1}W]\partial_{x_1} W\,.
\end{align*}
The first term is the main term
\begin{align*}
\scalar{\partial_{x_1}W}{\e^{\i\Delta t}\partial_{x_1}(-\Delta)^{-1}W}=\frac{1}{3}\scalar{W}{\e^{\i\Delta t}W}=-\frac{1}{3\sqrt{2}}\pi^{\pdz}t^\mdz+O(t^{-\frac{5}{2}})\,
\end{align*} which is obtain by Fourier transformation and integration by parts, see ~\cite{Froehlich102}.
For the various other terms on the right hand side we use the propagator estimates in Proposition ~\ref{prop:propagator} to estimate:
\begin{align*}
g|\scalar{(h^{-1}W)\partial_{x_1} W}{\e^{\i\Delta t}\partial_{x_1}(-\Delta)^{-1}W}|\leq g t^\mdz
\end{align*} for the third
\begin{align*}
&g|\scalar{V}{\e^{\i\Delta t}\partial_{x_1}(-\Delta)^{-1}Wh^{-1}W}|\leq g t^\mdz\\
\end{align*} and for the last
\begin{align*}
g|\scalar{V}{\int_0^t\e^{\i\Delta(t-s)}W\e^{-\i hs}\partial_{x_1}h^{-1}W}|&\leq g \norm{\x^3V}_2\int_0^t(1+t-s)^\mdz\norm{\x^3W\e^{-\i hs}\partial_{x_1}h^{-1}W}_2\d s\\
\lesssim&g\int_0^t(1+t-s)^\mdz(1+s)^\mdz\norm{\x^4W}_2\d s\lesssim gt^\mdz\,.
\end{align*}
This concludes the proof of (\ref{eq:help1}).
\end{proof}

\section*{Acknowledgments}
The authors would like to thank Professors J\"urg Fr\"ohlich and Israel Michael Sigal for suggesting this project and very useful discussions.

\providecommand{\bysame}{\leavevmode\hbox to3em{\hrulefill}\thinspace}
\providecommand{\MR}{\relax\ifhmode\unskip\space\fi MR }
\providecommand{\MRhref}[2]{%
  \href{http://www.ams.org/mathscinet-getitem?mr=#1}{#2}
}
\providecommand{\href}[2]{#2}

\end{document}